\pdfoutput=1
\documentclass[%
 reprint,
superscriptaddress,
 amsmath,amssymb,
 aps,
floatfix,
]{revtex4-2}
\usepackage{graphicx}
\usepackage{dcolumn}
\usepackage{bm}
\usepackage{amsthm}
\usepackage{commath}
\usepackage{mathtools}
\usepackage{graphicx}
\usepackage{bigints}
\usepackage[colorlinks=true, hyperindex, breaklinks, linkcolor=blue, urlcolor=blue, citecolor=blue]{hyperref} 
\usepackage{amsmath,amsfonts,amssymb,color,graphicx,xcolor,physics}
\usepackage{tikz}
\usepackage{float}
\makeatletter
\let\newfloat\newfloat@ltx
\makeatother
\usepackage{algorithm} 
\usepackage{algorithmicx}
\usepackage[noend]{algpseudocode}

\DeclarePairedDelimiter\floor{\lfloor}{\rfloor}

\usepackage{lineno}
\newcommand{\sandwich}[3]{\bra{#1}#2 \ket{#3}}
\newcommand{\sbkt}[1]{\left(#1\right)}

\newtheorem{theorem}{Theorem}
\newtheorem{lemma}{Lemma}
\newtheorem{definition}{Definition}
\newtheorem{corollary}{Corollary}
\newtheorem*{theorem*}{Theorem}
\newtheorem*{corollary*}{Corollary}

\newcommand\nonpfrate[1]{\gamma_{X, Y}}
\makeatletter
\newcommand*{\rom}[1]{\expandafter\@slowromancap\romannumeral #1@}
\makeatother

\let\oldproof\proof
\renewcommand{\proof}{\oldproof}

\def\algbackskip{\hskip-\ALG@thistlm}

\begin{document}

\preprint{APS/123-QED}

\title{Performance and achievable rates of the Gottesman-Kitaev-Preskill code for pure-loss and amplification channels}

\author{Guo Zheng} 
\thanks{These authors contributed equally.}
\email{guozheng@uchicago.edu}
\affiliation{Pritzker School of Molecular Engineering, The University of Chicago, Chicago 60637, USA}

\author{Wenhao He}
\thanks{These authors contributed equally.}
\affiliation{Center for Computational Science and Engineering, Massachusetts Institute of Technology, Cambridge, MA 02139, USA}

\author{Gideon Lee}
\affiliation{Pritzker School of Molecular Engineering, The University of Chicago, Chicago 60637, USA}

\author{Kyungjoo Noh}
\affiliation{AWS Center for Quantum Computing, Pasadena, CA 91125, USA}

\author{Liang Jiang}
\email{liang.jiang@uchicago.edu}
\affiliation{Pritzker School of Molecular Engineering, The University of Chicago, Chicago 60637, USA}
\date{\today}

\begin{abstract}

    Quantum error correction codes protect information from realistic noisy channels and lie at the heart of quantum computation and communication tasks. Understanding the optimal performance and other information-theoretic properties, such as the achievable rates, of a given code is crucial, as these factors determine the fundamental limits imposed by the encoding in conjunction with the noise channel. Here, we use the transpose channel to analytically obtain the near-optimal performance of any Gottesman-Kitaev-Preskill (GKP) code under pure loss and pure amplification. We present rigorous connections between GKP code's near-optimal performance and its dual lattice geometry and average input energy. With no energy constraint, we show that when $\abs{\frac{\tau}{1 - \tau}}$ is an integer, specific families of GKP codes simultaneously achieve the loss and amplification capacity. $\tau$ is the transmissivity (gain) for loss (amplification). Our results establish GKP code as the first structured bosonic code family that achieves the capacity of loss and amplification. 

\end{abstract}

\maketitle


\section{Introduction}

An ultimate question in the theory of communication~\cite{nielsen_chuang_2010, wilde2017quantum} is to determine the maximum information transmission rate, also known as capacity, of various communication models and channels. Extending the seminal works of Shannon~\cite{Shannon_1948_1, Shannon_1948_2} to channels of quantum mechanical nature, quantum Shannon theory has seen significant developments in recent years~\cite{PhysRevLett.98.130501, Holevo+2013, RevModPhys.66.481, RevModPhys.84.621}. As an example of the capacity quantities, the (one-way) quantum capacity~\cite{PhysRevA.54.2629, PhysRevA.55.1613, 1377491, PhysRevA.57.4153} focuses on the communication model where no classical communications are allowed, and it is defined to be the supremum of the quantum error correction (QEC)~\cite{shor1997faulttolerant, PhysRevLett.78.2252} coding rate with perfect decoding of the transmitted information. Straightforwardly, the achievable rates of specific codes provide lower bounds to the quantum capacity. However, computing the rates through the original definition is often cumbersome: it is highly nontrivial to provide a decent guess on both the encoding and decoding simultaneously. Moreover, since the coding rate is defined over asymptotically many channel uses, the performance of the encoder-decoder pair would need to be obtained analytically. Alternatively, it is more common to calculate achievable rates through the coherent information~\cite{1377491, PhysRevA.55.1613, Shor_lecture_notes}. The advantage of this method is that it only requires the conjecture of an input state. Many state-of-the-art bosonic channels' capacity lower bounds are determined through this approach~\cite{PhysRevA.63.032312, Wolf_2007, Noh_2020}. Among them, it was shown that loss and amplification channel capacities are exactly known and coincide with the coherent one-shot information of thermal input states~\cite{PhysRevA.63.032312, Wolf_2007, RevModPhys.84.621, PhysRevA.74.062307}. Nevertheless, a drawback of this method is that it only proves the existence of an optimal code based on the input states~\cite{1377491, PhysRevA.55.1613, Shor_lecture_notes, Hayden_2008_decouple}. The code construction has little structure, which poses great difficulty in its decoding and practicality.

The Gottesman-Kitaev-Preskill (GKP) code~\cite{PhysRevA.64.012310} is a promising bosonic code that has attracted active theoretical and experimental effort in superconducting circuits~\cite{Campagne_Ibarcq_2020, Sivak_Nature_2023, PhysRevLett.132.150607}, ion traps~\cite{deNeeve_2022, Fl_hmann_2019}, and optical platforms~\cite{PhysRevA.102.012607, Bourassa2021blueprintscalable}. Its popularity is largely due to the Gaussian nature of its logical operations~\cite{PhysRevA.64.012310} and the recent advances in its decoding~\cite{PhysRevLett.125.260509, Hastrup:2021aa}. Another key motivation is that while the GKP code was originally designed against displacement noise~\cite{PhysRevA.64.012310}, there has been numerical evidence~\cite{PhysRevA.97.032346, 8482307, Leviant_2022} that demonstrates the superiority of the optimally decoded single-mode GKP code under loss compared to other single-mode bosonic codes. Compared to displacement noise, loss is more experimentally motivated and is the dominant noise process in various hardware platforms such as microwave cavities~\cite{PhysRevX.6.031006} and optical fibers~\cite{RevModPhys.84.621}. The outstanding performance of GKP code under loss is rather surprising. While there have been some qualitative arguments~\cite{PhysRevA.97.032346, PhysRevA.93.012315}, there is not yet a rigorous and quantitative explanation that connects GKP code's performance with its parameters, such as its dual lattice and energy. In addition, the optimal performance of GKP code under loss is obtained through numerical convex optimization~\cite{PhysRevA.97.032346, 8482307}, leading to restrictions on the accessible system size. It is left unknown whether GKP code can maintain its performance at larger average photon numbers.

The information-theoretic properties of GKP codes are also remarkable. Since its code construction is based on translationally-invariant lattices~\cite{PhysRevA.64.012310, conrad2024good, Conrad_2022, PRXQuantum.3.010335, PRXQuantum.4.040334}, the performance of its optimal decoder against displacement noise is tightly connected to the geometry of its dual lattice. The achievable rate of GKP code under Gaussian displacement noise matches the lower bound of the channel's quantum capacity~\cite{PhysRevA.64.062301}. It is tempting to extend the analysis to GKP code under loss. Nevertheless, the analysis of GKP code's optimal decoder against loss is cumbersome to analyze. Taking a step back, past works~\cite{PhysRevA.97.032346, 8482307} obtained a nonvanishing achievable rate through the amplification decoder (AD), whose performance is tractable yet suboptimal. The achievable rate of the GKP code under such a decoder is shown to have a nonzero gap compared to the capacity of the loss channel. It is unclear whether the gap is due to the non-optimal decoder, the limitations of the GKP code, or a combination of both.

In this article, we obtain the near-optimal performance of GKP codes under loss and amplification. The key tool we adopt is the recently obtained formalism of near-optimal fidelity~\cite{zheng2024near}. The near-optimal fidelity provides an easy-to-compute metric based on the QEC matrix and is guaranteed to be close to the optimal decoder performance. With the near-optimal fidelity, we are able to obtain analytical expressions for the GKP code's performance, which overcomes the system size restrictions of the optimization-based approaches. The expression reveals the connection between GKP code's performance and its dual lattice geometry. As an example, we numerically compare the optimal performance of finite-energy square-lattice GKP code with the performance of multiple existing decoders and observe a performance gap at order(s) of magnitude. At infinite energy, we show that when $\abs{\frac{\tau}{1 - \tau}}$ is an integer, GKP codes based on specific lattice families achieve the capacity of not only loss but also amplification. Here, $\tau = \eta\leq 1$ is the transmissivity for loss, and $\tau = G\geq 1$ is the gain for amplification. The lattice families can be, e.g., scaled symplectic self-dual lattices or NTRU-based lattices~\cite{Harrington2004AnalysisOQ, conrad2024good}. To the best of our knowledge, our results establish GKP code as the first structured bosonic code family that achieves the capacity of loss and amplification. Last but not least, our methodology based on the near-optimal fidelity only requires the conjecture of the input QEC code. Thus, it opens the door to computing the achievable rates of other multimode bosonic codes~\cite{Jain_2024} or qubits codes.

\section{Results}


\subsection{The Gottesman-Kitaev-Preskill code}

The GKP code~\cite{PhysRevA.64.012310} is a lattice-based bosonic code. Specifically, the stabilizer group of an infinite-energy GKP code corresponds to translations along the lattice vectors of a symplectic integral lattice. There are two key parameters of a GKP code, i.e. the underlying lattice, $\Lambda$, and the envelope parameter, $\Delta$. $\Delta$ is also referred to as the effective squeezing and determines the energy of the code. An infinite-energy ($\Delta = 0$) GKP code has codewords $\ket{\mu_L}_0$, stabilized by
\begin{eqnarray}
    \hat{S}(\mathbf{v}) = \hat{T}(\sqrt{2\pi}\mathbf{v}),
\end{eqnarray}
for all lattice points $\mathbf{v}\in \Lambda$. Here, the translation operator $\hat{T}(\mathbf{u}) := \exp{-i\mathbf{u}^T \Omega \hat{\mathbf{x}}}$, and $\Omega$ is the symplectic form. Finite- and infinite-energy GKP codes are connected through the envelope operator as
\begin{eqnarray}
    \ket{\mu}_\Delta = N_\mu e^{-\Delta^2 \hat{n}}\ket{\mu}_0,
\end{eqnarray}
where $N_\mu$ is the normalization constant, and $\hat{n}$ represents the number operator. While the above defines single-mode GKP codes, most of them carry over straightforwardly to multimode GKP codes.

The symplectic dual lattice of any lattice is defined as 
\begin{eqnarray}
\Lambda^\perp := \set{\vec{u}\vert \vec{u}^T\Omega \vec{v}\in \mathbb{Z}, \forall \vec{v}\in\Lambda}.
\end{eqnarray}
The commutation relation of the stabilizer group is guaranteed by $\Lambda\subseteq \Lambda^\perp$, a property known as symplectically integral or weakly self-dual. The logical operators of the GKP code correspond exactly to translations along $\mathbf{v}\in\Lambda^\perp\notin \Lambda$. Therefore, when the GKP code is analyzed against random displacement errors, there is a natural notion of distance given by the closest lattice point~\cite{PRXQuantum.4.040334, Conrad_2022} in the dual lattice
\begin{eqnarray}
    \abs{\Lambda^\perp}_{\text{min}} := \min_{\mathbf{v}\in\Lambda^\perp} \abs{\mathbf{v}}.\label{eq:dual_lattice_closest_distance}
\end{eqnarray}
While the focus of this work is on GKP code under more practical noise channels like loss, the same distance, surprisingly, plays a critical role.

\subsection{Quantum capacities and achievable rates}

\begin{figure*}
    \centering
    \includegraphics[width = 0.9\textwidth]{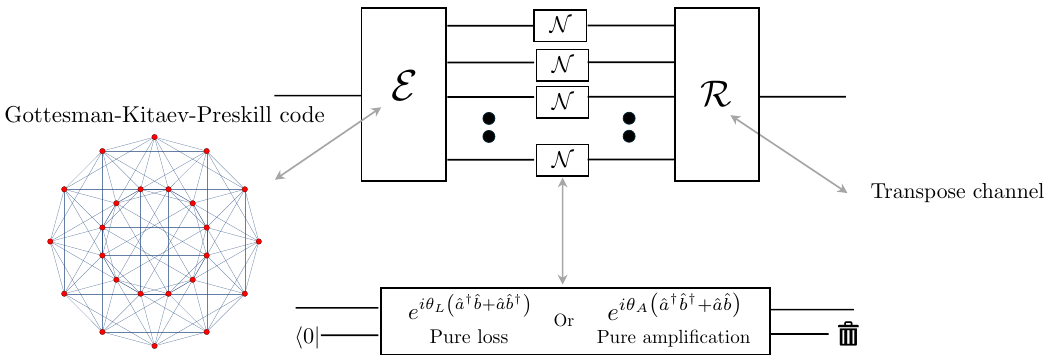}
    \caption{Coding scheme for one-way quantum information transmission without classical assistance. The input is any state within a $d_L$-dimensional logical space, which transmits through $N$ copies of the communication channel $\mathcal{N}$. In the end, the recovery decodes the logical information. In this work, we focus on $\mathcal{N}$ being either pure loss or pure amplification, whose dilated dynamics correspond to beam-splitter and two-mode squeezing respectively. The encoding is a Gottesman-Kitaev-Preskill code, whose near-optimal performance is tightly connected to the packing properties of its underlying lattice. The shown plot of the 24 cell corresponds to the Voronoi cells of the D4 lattice, which is the densest lattice packing in 4-dimensional space. The transpose channel is applied as the near-optimal recovery.}
    \label{fig:illustrative}
\end{figure*}

In this work, we focus on the communication model of transmitting quantum information through quantum channels without classical communications, also known as one-way communication~\cite{wilde2017quantum}. Suppose we are interested in transmitting information through channel $\mathcal{N}$, we have the option to perform arbitrary collective encoding, $\mathcal{E}$, of the information into $N$ copies of the physical systems. After the encoding, each system passes through an independent application of $\mathcal{N}$. At the end, a recovery, $\mathcal{R}$, restores and decodes the information stored. This procedure is illustrated in Fig.~\ref{fig:illustrative}. The performance of $\mathcal{E}$ and $\mathcal{R}$ is parameterized by the fidelity $F$ and the encoded logical dimensions $d_L$. Throughout our work, we adopt the channel fidelity as the fidelity metric, which is defined for any channel $\mathcal{Q}$ as~\cite{PhysRevA.75.012338, PhysRevA.54.2614}
\begin{eqnarray}
    F\left(\mathcal{Q}\right) := \bra{\Phi} \mathcal{Q} \otimes\mathcal{I}_A\left(\ket{\Phi}\bra{\Phi}\right)\ket{\Phi}
\end{eqnarray}
where $\ket{\Phi}$ is the purified maximally mixed state, and $\mathcal{I}_A$ is the identity channel acting on the \emph{ancillary} system. The channel fidelity characterizes how well entanglement is preserved~\cite{Tomamichel_2016} and has tight connections with metrics like average input-output fidelity~\cite{PhysRevA.60.1888, Nielsen_2002}. Generally speaking, there exists a tradeoff relation between $F$ and $d$: as the logical dimension increases, the redundancy in the physical system decreases, providing less protection against the loss of information. Conveniently, the achievable rate combines the above and serves as a unified benchmark~\cite{Wang_2022}, defined as 
\begin{eqnarray}
    R(\mathcal{E},\mathcal{R}) := \lim_{N\to\infty}\frac{1}{N}\log_2 d_L(\mathcal{E},\mathcal{R}), \text{ as } 1-F\to0.
\end{eqnarray}
Intuitively, the achievable rate is the maximum number of logical qubits transmitted per channel used in the asymptotic limit while the encoded information can still be perfectly decoded. The (one-way) quantum capacity follows as 
\begin{eqnarray}
        C_Q := \sup_{\mathcal{E}, \mathcal{R}} R(\mathcal{E},\mathcal{R}),
\end{eqnarray}
which sets an upper bound on all possible achievable rates. Thus, by definition, if a code's rate achieves the channel's capacity, the code is one of the optimal codes under the given channel. However, this definition has not been very useful in determining the capacities since it involves optimization over both encodings and recoveries. As an alternative, a seminal result~\cite{1377491, PhysRevA.55.1613, Shor_lecture_notes} on quantum capacity states that
\begin{eqnarray}
    C_Q = \lim_{N\to\infty}\max_{\rho\in \mathcal{D}(\mathcal{H}_0^{ N})}\frac{1}{N}I_c(\hat{\rho}, \mathcal{N}^{\otimes N})\label{eq:coherent_information_bound}
\end{eqnarray}
where $\mathcal{H}_0^{N}$ is the Hilbert space $\mathcal{N}^{\otimes N}$ acts on and $I_c$ is the coherent information, defined in Methods. Eq.~\eqref{eq:coherent_information_bound} is a direct consequence of the existence of an encoder-decoder pair that can reliably transmit information at a rate equal to the coherent information. To compute the coherent information, one only needs to specify an input state, which provides a convenient approach to lower bound the capacity. Many bosonic channel capacity lower bounds are obtained in this manner, such as thermal states for pure loss, pure amplification, and Gaussian random displacement channels~\cite{PhysRevA.63.032312}. While most bosonic channels' capacities are not fully determined, the pure loss and pure amplification channels are amongst the few with known capacities~\cite{PhysRevA.63.032312, RevModPhys.84.621, Lami_2023}, thanks to them being degradable channels~\cite{wilde2017quantum}. Both channels can be understood through a dilated form~\cite{PhysRevLett.101.200504} $\mathcal{N}(\rho) = \Tr_E \left(\hat{U} \left(\rho\otimes \ket{0}\bra{0}\right)\hat{U}^\dagger\right)$, where the unitary corresponds to a beam-splitter and two-mode squeezing for pure loss and amplification, respectively. In practice, the environment would be a finite temperature bath, which gives rise to thermal loss or thermal amplification. Nevertheless, the environment temperature is generally small, and pure loss or amplification serves as a good approximation to the noise process. Henceforth, we refer to pure loss (pure amplification) as loss (amplification) for simplicity. These channels admit Kraus decompositions $\set{\hat{N}_l}$ as
\begin{eqnarray}
    \hat{N}_l = 
    \begin{cases}
        &\sqrt{\frac{1}{l!}\left(\frac{1-\eta}{\eta}\right)^l}\hat{a}^l\eta^{\hat{n}/2} \;\;\;\;\;\;\text{(Loss)}\\
        &\\
        &\sqrt{\frac{1}{l!}\frac{(G-1)^l}{G}}G^{-\frac{\hat{n}}{2}}\hat{a}^{\dagger l}\;\;\;\;\;\;\text{(Amp)},
    \end{cases}
\end{eqnarray}
and their channel capacities can be written as~\cite{RevModPhys.84.621}
\begin{eqnarray}
        C_Q = \max\left(\log_2\abs{\frac{\tau}{1-\tau}},  0\right),\label{eq:capacity}
\end{eqnarray}
where $\tau = \eta\leq 1$ for loss and $\tau = G\geq 1$ for amplification.

\subsection{The near-optimal fidelity}

The near-optimal fidelity~\cite{zheng2024near}, $\tilde{F}^{\text{opt}}$, benchmarks the near-optimal performance of a QEC code. Specifically, it provides a narrow two-sided bound over the optimal fidelity, $F^{\text{opt}}$, as 
\begin{eqnarray}
    \frac{1}{2}\left(1 - \tilde{F}^{\text{opt}} \right)\leq 1 - F^{\text{opt}} \leq 1 - \tilde{F}^{\text{opt}}.\label{eq:two_sided_bound}
\end{eqnarray}
Here, optimal fidelity is defined as the optimal performance over any completely positive trace-preserving recovery channels, which can be obtained through convex optimizations~\cite{PhysRevLett.94.080501, PhysRevA.65.030302, 8482307, PhysRevA.97.032346, PhysRevA.75.012338}. The key quantity towards computing the near-optimal fidelity is the QEC matrix, defined as
\begin{eqnarray}
    M_{[\mu  l], [\nu k]} = \bra{\mu_L} \hat{N}_l^\dagger \hat{N}_k\ket{\nu_L}, \label{eq:QEC_mat}
\end{eqnarray}
in index notation, where $\ket{\mu_L}$ are the logical codewords and $\hat{N}_l$ are the Kraus operators of the noise channel of interest. With the QEC matrix, the near-optimal fidelity can be computed through
\begin{eqnarray}
    \tilde{F}^{\text{opt}}=\frac{1}{d_L^2}\left\|\operatorname{Tr}_L \sqrt{M}\right\|_F^2,\label{eq:near_opt_F}
\end{eqnarray}
where the partial trace is taken over the code space indices, the matrix norm is the Frobenius norm, and $d_L$ is the logical dimension of the code. Therefore, the metric provides a computationally efficient method to obtain a rigorous approximation of any code's optimal performance. Importantly, the near-optimal fidelity can be achieved through the transpose channel~\cite{10.1063/1.1459754, Petz:1988usv, PhysRevLett.128.220502, PhysRevA.81.062342}, also known as the Petz recovery.

While Eq.~\eqref{eq:near_opt_F} is interesting in its own right, it is not straightforward to obtain analytical expressions from it. Inspired by the Knill-Laflamme condition~\cite{PhysRevA.55.900}, one can divide the QEC matrix into the correctable part, $I_L\otimes D$ and the residual uncorrectable part, $\Delta M := M - I_L\otimes D$. Here, $D$ is taken to be a diagonal matrix, which can be achieved either through truncation or through choosing the unitary gauge of the noise channel's Kraus representation. Qualitatively, $I_L\otimes D$ represents the effect of the noise that does not lead to ambiguity in the logical information. Therefore, expanding the matrix square root in Eq.~\eqref{eq:near_opt_F}, the near-optimal fidelity's perturbative expression can be written as~\cite{zheng2024near}
\begin{eqnarray}
    \tilde{\epsilon} := \frac{1}{d_L}\norm{f(D)\odot \Delta M }_F^2,\label{eq:F^TC_pert}
\end{eqnarray}
such that 
\begin{eqnarray}
    1-\tilde{F}^{\text{opt}} = \tilde{\epsilon} + O\left(\tilde{\epsilon}^{3/2}\right) \label{eq:F^TC_pert_approximation}
\end{eqnarray}
where we define $f(D) _{[\mu l],[\nu k]} = \frac{1}{\sqrt{D_{ll}} + \sqrt{D_{kk}}}$, and the Hadamard product $\left(A\odot B\right)_{ij} = A_{ij}B_{ij}$, i.e. element-wise multiplication. In Eq.~\eqref{eq:F^TC_pert_approximation}, the residual error is bounded under the assumption that the logical dimension, $d_L$, is a finite constant. It is important to note that this does not hold if one would like to consider the achievable rate of a code family in the asymptotic limit. Nevertheless, one can show that~\cite{zheng2024near} if $\Tr D = 1$,
\begin{eqnarray}
        \lim_{\tilde{\epsilon}\to 0} F^{\text{opt}} = \lim_{\tilde{\epsilon}\to 0}\tilde{F}^{\text{opt}} =1. \label{eq:pert_form_property}
\end{eqnarray}
This is particularly useful since it implies that to prove a rate is achievable, it is sufficient to show $\tilde{\epsilon}$ vanishes. As defined in Eq.~\eqref{eq:QEC_mat}, $M$ only depends on the codewords and noise Kraus operators. Thus, obtaining achievable rates through the near-optimal fidelity is relatively accessible since it only requires the encoder ansatz, and the decoder implicitly adopts a near-optimal decoder.

Throughout this work, our main focus is on the near-optimal fidelity. To avoid confusion, we denote all fidelity measures related to the near-optimal performance with a tilde. See Ref.~\cite{zheng2024near} for a detailed comparison between the near-optimal fidelity $\tilde{F}^{\text{opt}}$ and the optimal fidelity $F^{\text{opt}}$.

\subsection{Near-optimal performance of GKP codes}

\begin{figure*}
    \centering
    \includegraphics[width = 1.0 \textwidth]{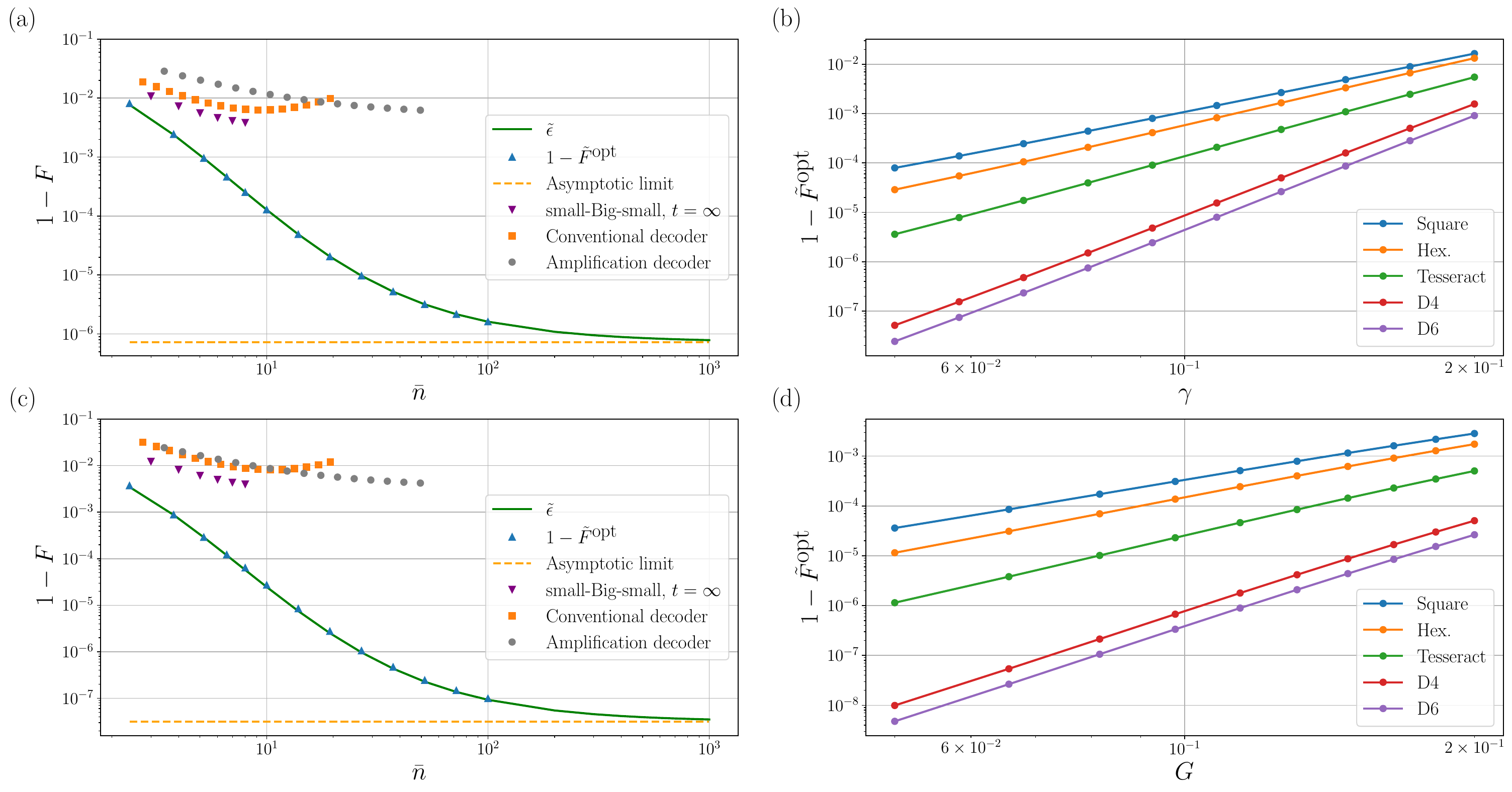}
    \caption{The performance of GKP codes under (a, b) loss and (c, d) amplification. (a, c) The infidelities of near-optimal recovery are compared against that of existing GKP decoders, such as AD~\cite{PhysRevA.97.032346, 8482307}, the conventional decoder~\cite{PhysRevA.64.012310}, and the sBs circuit~\cite{PhysRevLett.125.260509}. In particular, the sBs circuit is repeated until the system reaches a stabilized state. The noise channels are set to be (a) $\gamma = 10\%$ for loss and (c) $G = 1.1$ for amplification. (b, d) The performance of various lattices with parameters given in Tab.~\ref{tab:lattice_params} is presented. The encoding dimension is fixed to $d_L=2$ with $\bar{n}/N = 5$ average photon number per mode.}
    \label{fig:fidelities}
\end{figure*}

To analytically obtain the near-optimal performance of the GKP code, the QEC matrix defined in Eq.~\eqref{eq:QEC_mat} is the key ingredient. While it is in general tedious to compute the QEC matrix entries, it is convenient when we utilize the code structures. In this section, we focus on single-mode GKP codes with arbitrary lattice shapes for simplicity, but the results can be generalized to multimode GKP code. The QEC matrix for finite-energy GKP code involves three elements, the infinite-energy codewords $\ket{\mu}_0$ (see Supplementary), the envelope operator $e^{-\Delta^2 \hat{n}}$, and the combined Kraus operators $\hat{N}_{l}^\dagger \hat{N}_{k}$. The key observation is that the infinite-energy GKP states correspond to translationally-invariant lattices. Thus, displaced codewords are simply translated lattices, and their overlaps are straightforward to compute. To utilize such an observation, we can express all relevant operators with their displacement operator decompositions, i.e. their characteristic functions. For example, the Kraus operator combinations can be written as~\cite{SM, PhysRevA.97.032346}
\begin{eqnarray}
    &&\hat{N}_l^\dagger \hat{N}_{k}\nonumber \\
    &=& \int \frac{d^2 \alpha}{\pi} e^{-\frac{1}{2}\tau|\alpha|^2}\left\langle l\left|\hat{D}\left({\alpha^{\ast}}\right)\right| k\right\rangle \hat{D}\left(\alpha \sqrt{\abs{1-\tau}}\right)\label{eq:kraus_op_comb_noise_channel}
\end{eqnarray}
where $\ket{l}$ are Fock states with $l$ excitations, and $\tau = \eta\leq 1$ and $\tau = G\geq 1$ for pure loss and amplification respectively. The almost identical characteristics functions of these two noise channels' Kraus operator combinations are crucial and enable us to extend our result from one to the other.

In the following, we focus on single-mode GKP code under loss to highlight the key takeaways. The full expressions for multimode GKP code under loss or amplification are derived in the Supplementary Notes. For a GKP code with $d_L$ logical dimensions and lattice $\Lambda$ undergoing loss with probability $\gamma=1-\eta$, its QEC matrix expression can be written as
\begin{widetext}
    \begin{align}
        M_{[\mu, l],[\nu, k]} = \left(1 + \delta\right) \frac{\sqrt{R}^{l+k}}{\gamma n_\Delta + 1} \sum_{n_1,n_2\in \mathbb{Z}} e^{i\pi  n_{2}\left(n_{1}+\frac{\mu + \nu}{d_L}\right)}e^{-\frac{\pi}{2}\frac{1-\gamma}{ \gamma+\frac{1}{n_{\Delta}}}\left|\bold{L}\right|^2}\left\langle l\left|\hat{D}\left({\sqrt{\frac{\pi\left(n_\Delta + 1\right)}{\gamma n_\Delta + 1}}\left(L_1 - iL_2\right)}\right) \right| k\right\rangle\label{eq:QEC_matrix}
    \end{align}
\end{widetext}
where $\delta = O\left(e^{-\pi\left(n_\Delta + \frac{1}{2}\right)\abs{\Lambda^\perp}_{\text{min}}^2}\right)$, $t:= \frac{\gamma n_\Delta}{\gamma n_\Delta + 1}$, and $n_\Delta := \frac{1}{e^{2\Delta^2} -1}$. Here, we have defined the symplectic dual lattice to be $\Lambda^\perp$, and $\abs{\Lambda^\perp}_{\text{min}}$ is defined in Eq.~\eqref{eq:dual_lattice_closest_distance}. $\delta$ exponentially decreases with the input energy and is the direct consequence of the GKP codewords being non-orthogonal for finite $\Delta$. Suppose $M^\perp$ generates $\Lambda^\perp$, the vector $\bold{L}:= \left(M^\perp\right)^T \left(d_L n_1 + \mu -\nu, n_2\right)^T$ enumerates lattice points over a sublattice of the dual lattice such that 
\begin{eqnarray}
    \Lambda^\perp = \bigcup_{\mu - \nu \in \mathbb{Z}_{d_L}} \set{\bold{L}\vert n_1,n_2\in \mathbb{Z}}.\label{eq:dual_sublattice}
\end{eqnarray}
The appearance of the symplectic dual lattice is natural since we are working with displacement operators, and the displacements that exactly connect infinite-energy codewords are the ones that correspond to dual lattice points. 

Physically, $n_\Delta$ is tightly connected to GKP code's average energy. Suppose $\hat{P}_L$ is the codespace projection, and the average energy is defined as $\bar{n}:= \frac{1}{d_L}\Tr{\hat{P}_L\hat{n}}$, and it can be showed that for any GKP code~\cite{SM}
\begin{eqnarray}
    \bar{n} = \frac{1}{e^{2\Delta^2} -1} + \delta\approx n_\Delta.\label{eq:average_photon_number}
\end{eqnarray}
Since the gap of $\delta$ is exponentially suppressed with $\bar{n}$, it is already negligible for GKP codes with a few photons. The commonly cited expression of $\bar{n}\approx\frac{1}{2\Delta^2}-\frac{1}{2}$~\cite{PhysRevA.64.012310, PhysRevA.97.032346} can be seen as a taylor expansion of Eq.~\eqref{eq:average_photon_number} in small $\Delta$ limit.

With a different motivation, Ref.~\cite{PhysRevA.97.032346} derived the QEC matrix expression for single-mode square lattice GKP qubit code. While the general form is similar, their derivation is restricted to square lattices and assumed $\Delta\ll 1$ to simplify the derivations along the way. In contrast, the only approximation we made is in the prefactor of $\left(1+\delta\right)$. Moreover, our derivation is valid for arbitrary multimode lattices.

With the QEC matrix expression, we arrive at the near-optimal performance of the GKP code through the perturbative infidelity. The full infidelity expression contains summations over all dual sublattices given in Eq.~\eqref{eq:dual_sublattice}. Nevertheless, the leading order contributions are given by the closest lattice points of the dual lattice. Only leaving those leading terms, we arrive at
\begin{eqnarray}
    \tilde{\epsilon} &\approx& \frac{m}{2} e^{-2\pi \left(n_\Delta + \frac{1}{2}\right)\abs{\Lambda^\perp}_{\text{min}}^2} \times\nonumber\\
    &&\left( 8\sum_{\Delta l=0}^\infty\frac{R^{\Delta l}}{(R^{\Delta l} + 1)^2}I_{2\Delta l}(z) - I_0(z) -1\right), \label{eq:finite_energy_single_mode_fidelity}
\end{eqnarray}
where $m$ denotes the kissing number (the number of closest lattice points with distance $\abs{\Lambda^\perp}_{\text{min}}$), $I$ denotes modified Bessel functions, and $z := \pi (n_\Delta + 1) \sqrt{R}\abs{\Lambda^\perp}_{\text{min}}^2$. Some examples of single-mode GKP codes that are of interest include the square and hexagonal lattice GKP codes, whose relevant parameters are provided in Tab.~\ref{tab:lattice_params}.

\begin{table*}
    \centering
    \begin{tabular}{p{2cm}<{\centering} p{2.7cm}<{\centering} p{3.5cm}<{\centering} p{2.5cm}<{\centering} p{2.5cm}<{\centering} p{1cm}<{\centering}}
        \hline \hline
        Lattice type & Number of modes, $N$ & $\abs{\Lambda^\perp}^2_{\text{min}}$ & Kissing number, $m$ & $\text{det}(\Lambda)$\\
         \hline 
         Square & 1 & 1 & 4 & 1\\
        Hexagonal & 1 & $2/\sqrt{3}$ & 6 & 1\\
        Tesseract & 2 & $1/\sqrt{2}$ & 8 & 4\\
        $D4$ & 2 & $\sqrt{2}$ & 24 & 1\\
        $D6$ & 3 & $1$ & 12 & 4\\
        $E8$ & 4 & $2$ & 240 & 1\\
        Leech & 12 & 4 & 196560 & 1\\
         \hline \hline 
    \end{tabular}
    \caption{Examples of symplectic integral lattices parameters~\cite{Florian_2004, PRXQuantum.3.010335}. The kissing number is the number of closest lattice points in the dual lattice. The determinant of the lattice is connected to the GKP code's logical dimension through $\text{det}(\Lambda) = d_0^2$.}
    \label{tab:lattice_params}
\end{table*}

While Eq.~\eqref{eq:finite_energy_single_mode_fidelity} and similar approximations closely follow the near-optimal fidelity, it does not establish a rigorous bound on it. Therefore, we derive a rigorous upper bound in the infinite energy limit, i.e. $n_\Delta\to \infty$
\begin{eqnarray}
        \lim_{\bar{n}\to \infty} \tilde{\epsilon} \leq \frac{1}{4}\sum_{\bold{x}\in \Lambda^\perp\notin\set{\bold{0}}}e^{-\pi\frac{1-\gamma}{\gamma}\left|\bold{x}\right|^2 } \label{eq:inf_energy_performance}
\end{eqnarray}
for any multimode GKP code. Here, the summation sums over all lattice points of the dual lattice except the origin. The equality sign is taken for single-mode codes. Since the GKP code's performance improves with energy, this serves as an asymptotic limit. For the detailed derivation, see Lemma 15 in the supplemental notes~\cite{SM}. 

As an example, if we consider the leading contributions from the closest lattice points, the single-mode square lattice GKP code's performance is given by $e^{-\frac{\pi}{d_L}\frac{1-\gamma}{\gamma}}$. Such scaling can be compared with the performance of GKP code with AD, which gives an infinite-energy performance of $e^{-\frac{\pi}{4d_L}\frac{1-\gamma}{\gamma}}$~\cite{PhysRevA.97.032346} for square lattice. Our infidelity expression for near-optimal recovery shows a significant improvement over that of AD -- increasing the exponent by a factor of 4. This enhancement is not only practically important for high-performance bosonic QEC, but also crucial for our later proof that the GKP code achieves quantum channel capacity.

When the GKP code is considered in the context of displacement noise, the connection between the dual lattice and the GKP code's performance is apparent. However, to the best of our knowledge, this is the first time GKP's performance under loss (or amplification) has been rigorously connected to its dual lattice. The dual lattice spacings and the relevant phase have their origins in the QEC matrix expression since we decomposed the noise channel Kraus operators as displacement operators.

In Fig.~\ref{fig:fidelities}, we explicitly compare the performance of GKP decoders and lattices under loss and amplification. In subfigures (a, b) and (c, d), the noise channels are loss channels with loss probability $\gamma = 0.1$ and amplification channels with gain $G = 1.1$, respectively. In (a, c), we focus on the square lattice GKP code encoding $d_L=2$ logical dimensions. We present a comparison between the analytical and perturbative expression given in Eq.~\eqref{eq:finite_energy_single_mode_fidelity} (green lines) and the \textit{exact} values of the near-optimal infidelity, $1-\tilde{F}^\text{opt}$ (blue triangles). The exact values are computed through Eq.~\eqref{eq:near_opt_F}. Clearly, the perturbative expression approximates the exact values very well. Moreover, the asymptotic limits are shown in orange dashed lines, which performs surprisingly well: even with $10\%$ loss, GKP code can still achieve infidelities of order $10^{-6}$.

Numerically, one can observe that GKP code's performance improves monotonically as the average energy increases~\cite{zheng2024near}, which is distinct from other bosonic codes: there has been the misconception that it is more difficult to restore the lost information when the number of lost photons increase. Past results have found an optimal photon number for most other bosonic codes~\cite{PhysRevA.97.032346, PhysRevX.6.031006, PhysRevLett.119.030502}, and the GKP code is a first instance that does not have such a limitation. There could exist other unexplored bosonic codes that have similar features as well, such as the squeezed cat code~\cite{XuZheng2023, Schlegel_2022, Hillmann_2023}, which is a simplified GKP code with only a few peaks.

While the near-optimal fidelities reveal the limit of GKP code, we examine the performance of a few existing GKP code recoveries. For example, proposed along with the GKP code~\cite{PhysRevA.64.012310}, the conventional decoder is designed against displacement errors and recovers the information through modular quadrature measurement and adaptive displacement recovery. As shown in Fig.~\ref{fig:fidelities} (a, c), when it is applied against loss or amplification noise, its performance does not always improve with increasing photon numbers. Small-Big-small (sBs)~\cite{PhysRevLett.125.260509} is arguably the state-of-the-art recovery against loss, and it is widely adopted in superconducting circuit~\cite{Campagne_Ibarcq_2020, Sivak_Nature_2023, Eickbusch_2022} and ion trap experiments~\cite{Fl_hmann_2019, deNeeve_2022} for the correction and preparation of GKP states. A round of sBs consists of a short circuit comprised of conditional displacements and ancilla qubit rotations. As an optimistic benchmark, we show its stabilized fidelity after multiple rounds of sBs. Last but not least, AD is the decoder that established the state-of-the-art rate of GKP code under pure and thermal loss~\cite{8482307}. At the cost of injecting additional noise, it converts loss or amplification channels into Gaussian random displacement channels, whose optimal decoder at infinite energy is known to be the conventional decoder. The simulation details and exact definition of the recoveries are included in Methods.

 Overall, the comparisons between recoveries lead to the conclusion that in experimentally accessible regimes, e.g. $\bar{n}\sim 10$ and $\gamma = 10\%$, GKP code's performance can be improved by order(s)-of-magnitude solely through refining its recovery. The improvement can be even more significant for smaller $\gamma$. Since the near-optimal fidelity is achieved by the transpose channel, it is promising to approximate the transpose channel or other near-optimal recovery channels through techniques such as optimal control.

In Fig.~\ref{fig:fidelities} (b, d), we show the near-optimal infidelities that correspond to various lattices types. The average photon number per mode is fixed to $\bar{n}/N = 5$, and all codes encode $d_L = 2$ dimensional logical space. The lattices and their respective number of modes, $N$, shortest dual lattice vector sizes, and the kissing numbers are given in Tab.~\ref{tab:lattice_params}. In particular, it is apparent that as we increase the number of modes that encode the same amount of information, we effectively have more redundancy that we can use to protect the information, which is quantified by the distance of the dual lattice $\abs{\Lambda^\perp}_{\text{min}}$. Some lattices utilize space more effectively than others, which is related to the well-known sphere packing problem. For example, in two- and four-dimensional space, the tightest packing is achieved by the hexagonal and the $D4$ lattice respectively, rendering them more potential in correcting errors.
\begin{figure*}
    \centering
    \includegraphics[width = 1.0 \textwidth]{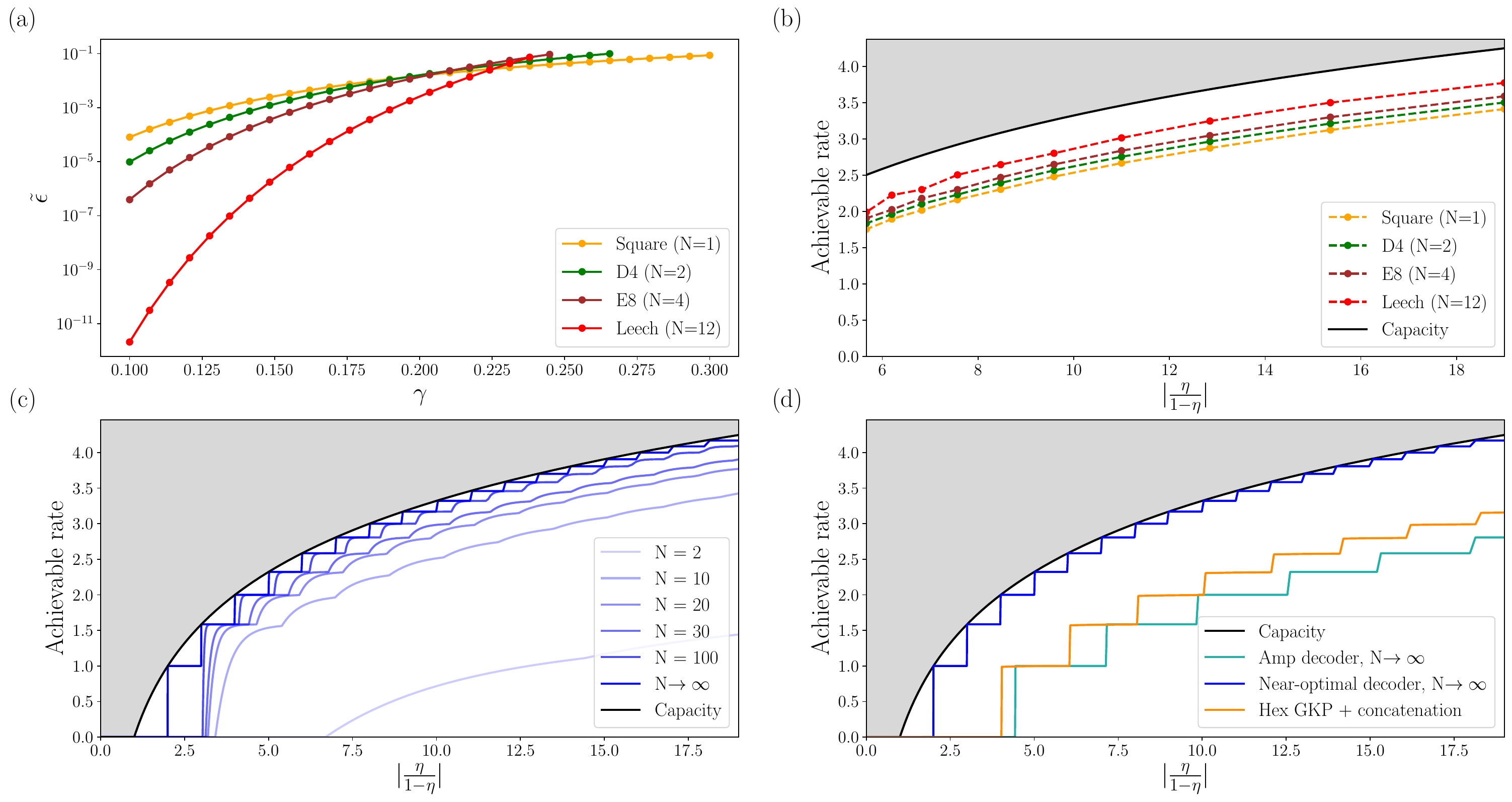}
    \caption{(a) The infidelities of GKP code are based on a few lattices with varying numbers of modes with a constant encoding rate such that $d_L = 3^N$. (b) Coherent information of GKP codes is based on various lattices, which are achievable rates. The coherent information is computed through the Werner states as given in Eq.~\eqref{eq:werner_coherent_info}. (c) Lower bounds on the achievable rates of GKP codes based on scaled symplectically self-dual lattices with an increasing number of modes, $N$. The infidelity is upper bounded through Eq.~\eqref{eq:upper_bound_infidelity_selfdual}. (d) A comparison of the rates achieved through the amplification decoder, the near-optimal decoder, and the single-mode hexagonal GKP code concatenated with random stabilizer codes for asymptotic system sizes. All subplots are focusing on infinite-energy GKP codes. \label{fig:rates}}
\end{figure*}

\subsection{Achievable rates}

In the previous section, we focused on the fidelity of the recovery, i.e. how well entanglement is preserved. Another equally important metric is the amount of logical information transmitted. In fact, the near-optimal fidelity derived has close connections with the logical dimension, $d_L$, as shown in Eq.~\eqref{eq:finite_energy_single_mode_fidelity}. In Fig.~\ref{fig:fidelities}(b,d), the comparison is made between various lattices transmitting the same amount of information, $d_L = 2$. However, this is not necessarily a fair comparison since some use more resources, e.g. number of modes, than others. In other words, the series of lattices we examine have a vanishing rate. In this light, we show a similar performance plot in Fig.~\ref{fig:rates}(a). The difference is that here we are showing the upper bound for the infinite-energy performance, as given in Eq.~\eqref{eq:inf_energy_performance}, and the logical dimensions are set to $d_L = 3^N$ with $N$ being the number of modes used for each lattice type. Recall that the quantum capacity of loss channel is $\log_2\frac{1-\gamma}{\gamma}$ number of qubits per mode. This essentially means that at this logical dimension rate, there exists an asymptotic code family that can transmit information arbitrarily well as we increase the number of modes if $\gamma<\frac{1}{d_L+1} = 0.25$. Qualitatively, Fig.~\ref{fig:rates}(a) supports the existence of such a threshold for the GKP code: when the loss probability gets sufficiently small, the fidelity of transmission improves with increasing system size, and vice versa. The crossover points of these lattices are not strictly $0.25$ for several reasons. Firstly, while the threshold is defined in the asymptotic limit and corresponds to certain good lattice families, the lattices shown are low-dimensional ones and are not guaranteed to be in the same optimal lattice family. Moreover, the performance crossing points increase towards 0.25 as we increase the number of modes, supporting the existence of a threshold. Last but not least, our performance metric is the near-optimal fidelity, which does not necessarily have the same crossover point as the optimal fidelity at finite error probability.

The achievable rate serves as a single benchmark that incorporates both the fidelity and the logical dimension. Other than following the original definition of the achievable rate, we can lower bound the achievable rates through coherent information. In particular, with a code of fidelity $F$ encoding $d_L$ dimensional logical space, it can be converted into a Werner state with coherent information~\cite{8482307}
\begin{eqnarray}
        && I\left(F, d_L\right) \nonumber\\
        &=& \log_2 d_L + F\log_2 F + (1-F)\log_2\frac{1-F}{d^2_L-1},\label{eq:werner_coherent_info}
\end{eqnarray}
whose regularized version, $I\left(F, d_L\right)/N$, is an achievable rate. In Fig.~\ref{fig:rates}(b), the achievable rates based on each base lattice are shown. For each base lattice, we compute their achievable rates with varying logical dimensions and plot the supremum. The logical dimensions are varied through scaling the base lattices. However, not all logical dimensions can be achieved through this approach. Assuming the base lattice has a logical dimension of $d_0$, we are only able to achieve logical dimensions of the form~\cite{PRXQuantum.3.010335, PhysRevA.64.062301, PhysRevA.64.012310}
\begin{eqnarray}
        d_L= \lambda^N d_0\label{eq:scale_log_dim}
\end{eqnarray}
where $\lambda$ is any positive integer and $N$ is the number of modes. See Methods for more details. Therefore, it is possible for the envelope of the achievable rates to be not smooth. Moreover, the rates are computed with $\tilde{\epsilon}$ as an approximation of $\tilde{F}^{\text{opt}}$. As shown in Eq.~\eqref{eq:F^TC_pert_approximation}, such an approximation only holds when the code performance is expected to be promising. Recall that the rates are optimized over all possible logical dimensions, we set the numerical constraint that we only optimize over logical dimensions such that $\tilde{\epsilon}\leq 10^{-2}$ to ensure the rates we compute are reliable.

By far, the shown lattices are specific instances and are not in the same asymptotic family. We can instead focus on a lattice family called the self-dual symplectic lattices~\cite{PhysRevA.64.062301, Harrington2004AnalysisOQ, PhysRevA.64.012310, Conrad_2022}. Such lattices have dual lattices identical to the prime lattices. A direct consequence is that the base lattice only encodes a one-dimensional space, i.e. $d_0=1$, and it can be scaled to encode $\lambda^N$ logical dimensions. A crucial property~\cite{PhysRevA.64.062301, busar_sarnak_1994} of the self-dual family is that there exists a self-dual lattice, $\Lambda$, such that 
\begin{eqnarray}
        \sum_{\bold{x}\in \Lambda\notin \set{\bold{0}}} f(\bold{x}) \leq \int f(\bold{x}) d^{2N} \bold{x}
\end{eqnarray}
for any integrable and rotationally invariant function $f$. As a result, for an infinite-energy GKP code that corresponds to a scaled self-dual lattice, $\sqrt{\lambda} \Lambda$, its infidelity is upper bounded through
\begin{eqnarray}
    \lim_{\bar{n}\to \infty} \tilde{\epsilon} \leq \frac{1}{4}\left(\lambda \frac{\gamma}{1-\gamma}\right)^N, \label{eq:upper_bound_infidelity_selfdual}
\end{eqnarray}
which vanishes if the dimension per mode $\lambda \leq \frac{1-\gamma}{\gamma} + \delta$ and $\delta\to0$. See Theorem 1 in supplemental notes for details~\cite{SM}. Recall the property of the near-optimal infidelity's perturbative form given in Eq.~\eqref{eq:pert_form_property}, the condition for a vanishing $\tilde{\epsilon}$ corresponds to the condition of perfect entanglement preservation. Following a similar derivation for the amplification channel, GKP codes achieve rates of 
\begin{eqnarray}
        R = \max\left(\log_2\floor*{\frac{\tau}{\abs{1-\tau}}},  0\right) \label{eq:multimode_rate}
\end{eqnarray}
where $\tau = \eta\leq 1$ for loss and $\tau = G\geq 1$ for amplificatio. Here, $\floor*{\cdot}$ denotes the floor function. In Fig.~\ref{fig:rates}(c), we show the rates that are achievable through different dimensions of the self-dual lattice. For a finite number of modes, such a bound would not give a vanishing error rate, but we can still obtain a rate through Eq.~\eqref{eq:werner_coherent_info}. It is clear that as we increase the number of modes we use, we can achieve superior rates with genuine multimode lattices compared to schemes based on low-dimensional lattices. For genuine lattices at $N\to\infty$, there is a step-like feature due to the floor function in Eq.~\eqref{eq:multimode_rate}. Such a feature appeared in relevant works where a similar construction was used to achieve the lower bound of the Gaussian random displacement channel or finite rate for loss~\cite{PhysRevA.64.062301, 8482307}. Therefore, the GKP code achieves the channel capacity for loss or amplification only when $\frac{\eta}{\abs{1-\eta}}\in \mathbb{Z}$. However, such a step-like feature is an artificial feature because of the limited choices of logical dimensions through scaling self-dual lattices. It is likely that with some alternative lattice constructions, the GKP code could always achieve the channel capacity. Another construction for symplectic lattices with goodness properties is the NTRU-based lattices~\cite{conrad2024good}, but it faces the same restriction on the logical dimension. There are algorithms that generate these lattice families with high probabilities~\cite{Harrington2004AnalysisOQ, conrad2024good}.

In the past, there have been multiple works on the achievable rates of GKP code. In particular, Ref.~\cite{8482307} established the state-of-the-art rate against loss channel before this work. The key to their approach is the amplification decoder (AD), which converts the (thermal) pure loss channel into a Gaussian random displacement channel through an amplification channel. However, the amplification channel inevitably introduces noise from the environment port of the amplification channel, leading to the suboptimality of AD. The rates of GKP code under AD are shown in Fig.~\ref{fig:rates}(d), which has a finite gap of $\log_2 e\approx 1.4$ compared to the capacity. Whether the gap was due to the deficiency in the recovery was left as an open question. Our result demonstrates that the GKP encoding with a (near-)optimal recovery is sufficient to close such a gap.

An alternative method to obtain rates from bosonic codes is to concatenate them with qudit codes. For example, on the lower level, we can select infinite-energy single-mode hexagonal GKP code. If we perform twirling, the noise channel can be understood as a generalized Pauli channel or a depolarizing channel, depending on the specific twirl we perform. The only constraint is that the summed error probability should equal $1-F$. To obtain a lower bound on the rates, we assume the channel is twirled such that we are dealing with a qudit depolarizing channel, i.e. the error probability vector is $\bold{p} =\left(F, \frac{1-F}{d_L^2 - 1}, \dots, \frac{1-F}{d_L^2 - 1}\right)$. The achievable rate of concatenating it with random stabilizer codes is given by the hashing bound (see Methods for more details) and is shown in Fig.~\ref{fig:rates}(d). Note that the hashing bound is strictly below the achievable rate of the genuine multimode lattices. As pointed out in Ref.~\cite{PhysRevA.64.062301}, while concatenated single mode GKP lattices can also be understood as lattices living in $2N$ dimensions, they do not achieve as good a packing as the sphere-packing lattices like the scaled self-dual lattices. 

One interesting conclusion from our approach based on scaled self-dual lattice is that there exists a GKP code that simultaneously achieves the capacity of loss and amplification given that 
\begin{eqnarray}
        \frac{1-\gamma}{\gamma} = \frac{G}{G-1}\label{eq:equivalence_loss_amp}
\end{eqnarray}
with tailored decoders. This observation is nontrivial and has its root in the almost identical form of loss and amplification Kraus operator combinations as given in Eq.~\eqref{eq:kraus_op_comb_noise_channel}. A relevant discussion is presented in the supplementary materials.

\section{Discussion}

In this work, we have studied the (near-)optimal performance of the GKP code under loss and amplification noise. The results are analytical and general such that they are applicable to multimode GKP codes over arbitrary input energy and lattice shapes. In particular, at experimentally accessible parameters, we presented a comprehensive comparison between the performance of existing decoders and a near-optimal decoder, i.e. the transpose channel. The limit of the GKP encoding is shown to be remarkably well, and it is possible to further suppress GKP code's error by order(s) of magnitude through refining the existing recoveries. Through exploiting the similar structures of the loss and amplification channels, we established the optimality of infinite-energy GKP code: it achieves the capacity for both channels simultaneously when $\abs{\frac{\tau}{1 - \tau}}$ is an integer. Here, $\tau$ is the transmissivity (gain) for loss (amplification). To the best of our knowledge, GKP code is the first known structured bosonic code family that achieves the capacity of bosonic channels that has a known capacity, i.e. pure loss and pure amplification.

This work applied a key technique called the near-optimal fidelity, which is based on the transpose channel, sometimes known as the Petz recovery~\cite{10.1063/1.1459754, Petz:1988usv, PhysRevLett.128.220502, PhysRevA.81.062342}. The Petz recovery has been known to have many information theoretic applications~\cite{PhysRevA.77.034101, 10.1116/5.0060893, PhysRevX.9.031029, Buscemi_2021}. Our work first demonstrated its power in determining a QEC code's achievable rate. It enables us to explore the performance and to determine the rates for codes with guaranteed near-optimality. Thus, it opens the door to computing the achievable rates of other multimode bosonic or qubit codes, such as the quantum spherical codes~\cite{Jain_2024}. In addition, while we have computed the rates of GKP codes under pure loss (amplification), its rates under thermal loss (amplification) remain unknown. With numerical evidence of GKP code's performance under thermal loss~\cite{8482307}, it would be very interesting to see if GKP code's rate saturates or even surpasses the state-of-the-art capacity lower bounds for these noisy channels.

\section{Methods}

\textbf{Coherent information.} Consider a quantum channel $\mathcal{N}$ and an input quantum state $\hat{\rho}$, the coherent information is given by 
\begin{eqnarray}
        I_c\left(\hat{\rho}, \mathcal{N}\right):= S\left(\mathcal{N}\left(\hat{\rho}\right)\right) - S\left(\mathcal{N}^c\left(\hat{\rho}\right)\right),
\end{eqnarray}
where $\mathcal{N}^c$ is the complementary channel, and $S$ denotes the von Neumann entropy. The (regularized) coherent information is tightly connected to the quantum channel capacity~\cite{1377491, PhysRevA.55.1613, Shor_lecture_notes} via Eq.~\eqref{eq:coherent_information_bound}. As a corollary, it has been established that for any coherent information, there exists a code that can achieve such a rate, which are lower bounds of the channel capacity. Consequently, given a state with a known channel fidelity, one can obtain a lower bound of the achievable rate by converting it into a Werner state after local operations and classical communications. For more details, see Refs.~\cite{8482307} and references therein.

Another tool we have applied is the qudit hashing bound. Suppose we consider a generalized Pauli channel as the noise channel in the Weyl operator basis, the hashing bound provides an achievable rate of random stabilizer codes and can be written as 
\begin{eqnarray}
        D_{\mathcal{N}} = (1-H_{d_L}(\bold{p}))\log_2 d_L
\end{eqnarray}
where $d_L$ is the number of levels of the qudits, $\bold{p}$ is the probability vector, $H_{n}$ is the entropy with base $n$, i.e. $H_n(\bold{x}):= -\sum_i x_i\log_n x_i$. $\log_2 d_L$ converts the unit from qudit to qubit. Worth noting, there are also other more general definitions of the Hashing bound (see, e.g., Refs.~\cite{Leviant_2022} and references therein) for general channels, but it converges to the expression above for qudit Pauli channels.

\textbf{The near-optimal channel fidelity and the transpose channel.} The near-optimal fidelity is a benchmark of a code's performance proposed in Ref.~\cite{zheng2024near}. It provides a tight two-sided bound of the optimal fidelity as given by Eq.~\eqref{eq:two_sided_bound}, where the fidelity metric we choose is the channel fidelity, also known as the process fidelity in some scenarios. The near-optimality of the near-optimal fidelity is a direct consequence of it being the performance of the transpose channel (TC) recovery, $\mathcal{R}^{\text{TC}}$, i.e.
\begin{eqnarray}
    \tilde{F}^{\text{opt}}:=F\left(\mathcal{R}^\text{TC}\circ \mathcal{N}\circ\mathcal{E}\right)
\end{eqnarray}
with $\mathcal{N}$ and $\mathcal{E}$ denoting the noise and encoding channels, respectively, and $\circ$ represents channel compositions. The transpose channel, also known as the Petz recovery in some circumstances, has been proven to be near-optimal, and possesses a constructive Kraus form,
\begin{eqnarray}
        \hat{R}^{\text{TC}}_i = \hat{P}_L \hat{N}_i^\dagger \mathcal{N}\left(P_L\right)^{-1/2}.\label{eq:TC_kraus_form}
\end{eqnarray}
Here, $\hat{N}_i$ are the Kraus operators of the noise channel, $\hat{P}_L$ is the logical codespace projector, and the inverse should be understood as the pseudoinverse. Therefore, the performance of any code that is given by its near-optimal fidelity is not only an existence proof of the recovery, but the explicit channel that realizes such a performance is also known.

\textbf{Basics of lattice theory and multimode GKP codes.} In this section, we give a brief introduction to lattice theory and its connection to GKP codes. For a more detailed explanation, we refer the readers to the supplementary materials and references therein. 

All $N$-mode GKP codes are in one-to-one correspondence with lattices in $2N$ dimensions. A lattice, $\Lambda$, can be specified through its generator matrix
\begin{eqnarray}
    M:= \begin{pmatrix}
    \bold{v}_1^T\\
    \vdots\\
    \bold{v}_{2N}^T\\
\end{pmatrix},
\end{eqnarray}
where $\bold{v}_i\in \mathbb{R}^{2N}$. The stabilizer group of the GKP code is exactly the translation operator that corresponds to the lattice point set. To have the stabilizers commute with each other, there is an additional constraint that the lattice is symplectically integral, also known as weakly self-dual, i.e. $\Lambda \subseteq \Lambda^\perp$. Here, $\Lambda^\perp$ represents the symplectic dual lattice of $\Lambda$. The self-dual lattices, which played a crucial role in our proof of the rates of the GKP code, are the special case where $\Lambda^\perp = \Lambda$.

The logical dimension of the GKP code is given by 
\begin{eqnarray}
        d_L = \text{det}(M)\in \mathbb{Z}.
\end{eqnarray}
One important approach to obtaining new lattices is to simply scale the base lattice in a uniform fashion, i.e. $M = \sqrt{\lambda} M_0$. Then, the scaled lattice can encode $d_L = \lambda^N d_0$. Therefore, when $d_0, d_L\in \mathbb{Z}$, we require $\lambda$ to be an integer as well, which leads to the conclusion that scaled lattices cannot achieve arbitrary logical dimensions, except when $N=1$ and $d_0 = 1$. To give some examples of the base lattices, the square and hexagonal lattices have generator matrices 
\begin{eqnarray}
    M_{\text{sq}} = \begin{pmatrix}
        1 & 0\\
        0 & 1
    \end{pmatrix},
    M_{\text{hex}} &=& \sqrt{\frac{2}{\sqrt{3}}}\begin{pmatrix}
        1 & 0\\
        -\frac{1}{2} & \frac{\sqrt{3}}{2}
    \end{pmatrix}.
\end{eqnarray}
For the generator matrices of other lattices, see Ref.~\cite{PRXQuantum.3.010335} and references therein.

\textbf{Recoveries of the GKP code.} In this section, we review the existing recoveries protocol that are compared against in the main text and the simulation details. To start with, the conventional decoder is the decoder proposed by the original work~\cite{PhysRevA.64.012310}. Consider a single-mode square lattice GKP code, the decoder measures the quadrature operators modular the lattice spacings and performs a corresponding displacement. It can be implemented through, e.g., a supply of ancillary GKP states along with SUM gates~\cite{PhysRevA.64.012310, PhysRevLett.125.080503}, phase estimations~\cite{PhysRevA.93.012315}, or teleportation-based error correction~\cite{PRXQuantum.3.010315}. However, while such a strategy is optimal for infinite-energy GKP code under displacement channel, it is far from optimal when we consider realistic noises since it does not take into account, for example, the contraction features of loss. 

A closely related decoder is the amplification decoder (AD)~\cite{PhysRevA.97.032346, 8482307}. The key idea is that through channel composition, we can convert loss (amplification) channel into Gaussian random displacement noise channel through amplifying (contracting) the logical states,
\begin{eqnarray}
        \mathcal{N}_L(\eta)\circ \mathcal{N}_A(G) = \mathcal{N}_{B_2}(\sigma^2)
\end{eqnarray}
with $\eta G = 1$ and $\sigma^2 = 1-\eta$. Here, $\mathcal{N}_L$, $\mathcal{N}_L$, and $\mathcal{N}_{B_2}$ represent pure loss, pure amplification, and Gaussian random displacements, respectively. After the noise channel is converted into Gaussian noise, it is natural to then append it with the conventional decoding. Nevertheless, the channel conversion process injects additional uncertainty into the system since it involves the participation of an input ancillary mode in a vacuum, which is eventually traced out. Therefore, it is also suboptimal. The simulations of both the conventional decoder and AD are performed through the Zak basis representation of the GKP code~\cite{Shaw_2024}, which is computationally efficient. 

The small-Big-small (sBs) is another recovery that is practical and has attracted experimental interest. The core advantage of it is its components are all experimentally accessible and are free from measurements, which could be the bottleneck for both operating speed and quality for many experimental platforms. More specifically, one round of sBs recovery circuit that stabilizes the position quadrature consists of a series of unitary operations on the oscillator and an ancillary qubit,
\begin{eqnarray}
    C\hat{D}\left(\delta/2\right)\hat{R}^\dagger (\pi/2)C\hat{D}\left(-il\right)\hat{R}(\pi/2)C\hat{D}\left(\delta/2\right)
\end{eqnarray}
where the conditional displacement is defined as $C\hat{D}\left(\alpha\right):= \ket{0}\bra{0}\otimes \hat{D}(\alpha/2) + \ket{1}\bra{1}\otimes \hat{D}(-\alpha/2)$ and $\hat{R}(\pi/2) := \text{exp}\left\{-i\hat{\sigma}_x\pi/4\right\}$. The displacement lengths $\delta \approx \Delta^2 l$ and $l = \sqrt{2\pi}$. The ancilla qubit is initialized in $\ket{+}$ and reset at the end. The sBs circuit is developed through the trotterization of the nullifier of GKP's finite-energy stabilizers. Since sBs does not stabilize the finite-energy GKP codespace with the Gaussian envelope, our simulation first discovers the stabilized codewords by initializing the oscillator in GKP codewords and applying the sBs circuits till convergence. Then, the stabilized codewords undergo the noise channel of interest and are later recovered through repeated applications of the stabilization circuits. The performance is evaluated through comparing the fidelities between the states before the noise channel and after the recovery.

\textbf{Approximate implementation of the near-optimal recovery.} While our core message is the analytical performance of the GKP code, there are also implications of practical interest. One of them is the possibility to design more optimal recoveries inspired by the near-optimal recoveries such as the transpose channel. 

Suppose one focuses on the two leading orders of the TC Kraus operators, i.e. $i =0,1$ in Eq.~\eqref{eq:TC_kraus_form}, it is straightforward to see that it roughly corresponds to a projection onto the no-error and one-error subspace and a rotation back to the codespace. One property of GKP code is that if it only encodes $d_L=2$~\footnote{Different from common believes~\cite{PhysRevX.10.011058}, this statement is not true for GKP encoding arbitrary logical dimensions.}, it possesses a 2-fold rotation symmetry~\cite{PhysRevX.10.011058} such that both of its codewords live in the even parity subspace. Therefore, take square lattice qubit GKP code as an example, one can approximate the error subspace projection through a parity measurement. Then, one can apply a unitary $\hat{U}_{0\to L}$ or $\hat{U}_{1\to L}$ depending on the parity being measured to be even or odd, respectively. Here, the unitary $\hat{U}_{i\to L}$ rotates the oscillator from the $i$-th error subspace back to the logical codespace. The exact implementation of such a unitary can be discovered through techniques such as optimal control. While we have been focused on the explicit implementation of the first two Kraus operators, it is worth noting that it is in principle possible to extend such analysis to higher orders or through systematic techniques that implement general CPTP maps~\cite{PhysRevB.95.134501}

\section*{Acknowledgments}
We thank Victor Albert, Jonathan Conrad, and Xiehang Yu for their helpful discussions. We acknowledge support from the ARO(W911NF-23-1-0077), ARO MURI (W911NF-21-1-0325), AFOSR MURI (FA9550-19-1-0399, FA9550-21-1-0209, FA9550-23-1-0338), DARPA (HR0011-24-9-0359, HR0011-24-9-0361), NSF (OMA-1936118, ERC-1941583, OMA-2137642, OSI-2326767, CCF-2312755), NTT Research, Samsung GRO,  Packard Foundation (2020-71479), and the Marshall and Arlene Bennett Family Research Program. This material is based upon work supported by the U.S. Department of Energy, Office of Science, National Quantum Information Science Research Centers, and Advanced Scientific Computing Research (ASCR) program under contract number DE-AC02-06CH11357 as part of the InterQnet quantum networking project. This work was completed with resources provided by the University of Chicago’s Research Computing Center.

\clearpage

\appendix 

\widetext
\begin{center}
\textbf{\large Supplemental Material}
\end{center}

\section{Basics on lattice theory}

In this section, we focus on the relevant lattice theory concepts that are helpful for our understanding of the GKP codes. In GKP codes, we are normally concerned with $N$ modes. Thus, we explore lattice theory in $2N$ dimensions. The lattice can be defined through a set of independent generators, $\bold{v}_i \in \mathbb{R}^{2N}$. These generators form the generator matrix,
\begin{eqnarray}\label{eq:gen_matrix}
    M:= \begin{pmatrix}
    \bold{v}_1^T\\
    \vdots\\
    \bold{v}_{2N}^T\\
\end{pmatrix}.
\end{eqnarray}
The linear combination of the lattice generators span the whole lattice, i.e.
\begin{eqnarray}
    \Lambda(M) := \set{M^T \bold{a} : \bold{a}\in \mathbb{Z}^{2N}}.
\end{eqnarray}
In particular, the type of lattices that are relevant to our discussion of GKP codes is the symplectic integral lattices, also known as weakly self-dual. Such lattices are defined to be lattices whose symplectic Gram matrix, defined as
\begin{eqnarray}\label{eq:sym_Gram_mat}
    A:= M\Omega M^T,
\end{eqnarray}
have only integer entries. The N-mode symplectic form is defined as $\Omega = I_N\otimes \omega$ and $\omega = \begin{pmatrix}
    0 & 1\\
    -1 & 0
\end{pmatrix}$. While a generator matrix uniquely determines the lattice, different generator matrices could correspond to the same lattice. 

\begin{lemma}[Canonical basis]
    For any symplectic integral lattice, there exists a canonical generator matrix $M$ such that its symplectic Gram matrix
    \begin{eqnarray}
    A= \operatorname{diag}\left(d_1, \cdots, d_{N}\right) \otimes \omega.
    \end{eqnarray}
\end{lemma}

The proof is omitted here for simplicity, but it can be found in, e.g., Ref.~\cite{PRXQuantum.4.040334} and references therein.

\begin{lemma}[Square lattice with a symplectic transformation]\label{lem:sq_lattice_connect}
    The canonical generator matrix of any symplectic integral lattice, $M$, can be connected with a square lattice through
    \begin{eqnarray}
    M = M_{\text{sq}} S^T
    \end{eqnarray}
    where $M_{\text{sq}} = D_{\text{sq}}\otimes I_2 = \operatorname{diag}\left(\sqrt{d_1}, \cdots, \sqrt{d_{N}}\right) \otimes I_2$, and $S^T = M_{\text{sq}}^{-1} M$ is a symplectic matrix.
\end{lemma}

Lemma~\ref{lem:sq_lattice_connect} is straightforward, and $S$ can be easily verified to be symplectic.

\begin{definition}[Symplectic dual lattice] For a lattice, $\Lambda$, its symplectic dual lattice is given by $\Lambda^\perp$ such that
    \begin{eqnarray} \label{eq:sym_dual_cond}
    \Lambda^\perp := \set{\vec{u}\vert \vec{u}^T\Omega \vec{v}\in \mathbb{Z}, \forall \vec{v}\in\Lambda}
    \end{eqnarray}
\end{definition}
Similarly, if the symplectic form in Eq.~\eqref{eq:sym_dual_cond} is absent, the lattice points form the Euclidean dual lattice.

\begin{definition}[Self-dual lattice]
    A symplectic self-dual lattice is a lattice, $\Lambda$, whose dual lattice is itself, i.e. $\Lambda = \Lambda^\perp$ and has a canonical symplectic Gram matrix of 
    \begin{eqnarray}
        A= I_N \otimes \omega.
    \end{eqnarray}
    where $I_N$ is a $N\times N$ identity matrix.
\end{definition}

\begin{lemma}[Existence of good self-dual lattice; \cite{busar_sarnak_1994}]\label{lem:Busar_good_lattice}
    There exist a $2N$-dimensional symplectic self-dual lattice, $\Lambda$, such that
    \begin{eqnarray}
            \sum_{\bold{x}\in \Lambda\notin \set{\bold{0}}} f(\bold{x}) \leq \int f(\bold{x}) d^{2N} \bold{x}
    \end{eqnarray}
    for any integrable and rotationally invariant function $f$.
\end{lemma}

\section{GKP code}

In this section, we focus on GKP and its properties. Following the conventions, we define the set of position and momentum operators for $N$ modes as
\begin{eqnarray}
    \hat{\bold{x}} := \left(\hat{q}_1, \hat{p}_1, \dots, \hat{q}_N, \hat{p}_N\right)^T.
\end{eqnarray}
There are two equivalent ways of expressing a displacement, being through $\bold{u}\in \mathbb{R}^{2N}$ and $\bold{\alpha}\in \mathbb{C}^N$
\begin{eqnarray}
    \hat{T}(\bold{u}) &:=& \exp{-i\bold{u}^T \Omega \hat{\bold{x}}},\\
    \hat{D}(\boldsymbol{\alpha}) &:=& \exp{\boldsymbol{\alpha} \hat{a}^\dagger - \boldsymbol{\alpha}^\ast \hat{a}},
\end{eqnarray}
and they are related by $\alpha_j = \frac{1}{\sqrt{2}}\left(u_{2j-1} + i u_{2j}\right)$. Equivalently, we can define a linear map
\begin{eqnarray}\label{eq:conversion_T_D}
    C:= \frac{1}{\sqrt{2}}\begin{pmatrix}
        1 & i & 0 & 0 & \dots\\
        0 & 0 & 1 & i & \dots\\
        && \dots &&
    \end{pmatrix},
\end{eqnarray}
such that $C:\mathbb{R}^{2N}\to \mathbb{C}^N$ for N modes, and if $C\bold{u} =\boldsymbol{\alpha}$, we have $\hat{T}(\bold{u}) = \hat{D}(\boldsymbol{\alpha})$. For any N-mode GKP, it has $2N$ independent stabilizer generators, each corresponding to a displacement,
\begin{eqnarray}
    \hat{S}_i = \hat{T}(\sqrt{2\pi}\bold{v}_i).
\end{eqnarray}
A generic element in the stabilizer group can then be written as
\begin{eqnarray}
    \hat{S} = \hat{T}(\sqrt{2\pi}  M^T \bold{a}),
\end{eqnarray}
where $M:= \begin{pmatrix}
    \bold{v}_1^T\\
    \vdots\\
    \bold{v}_{2N}^T\\
\end{pmatrix}$ and corresponds exactly to the generator of lattices as given in Eq.~\eqref{eq:gen_matrix}. Moreover, the stabilizers should commute with each other, leading to the condition that the symplectic Gram matrix shown in Eq.~\eqref{eq:sym_Gram_mat} should have only integer entries. Therefore, the stabilizer group of a GKP code is isomorphic to a symplectic lattice, with each stabilizer corresponding to a lattice point. The GKP's logical dimension is determined through the lattice as 
\begin{eqnarray}
    d_L = \text{det}(M)\in \mathbb{Z}.
\end{eqnarray}
As a result, if we scale the lattice through $M^\prime = \alpha M$, the new lattice should be symplectic integral, i.e. $\alpha^2 = \lambda \in \mathbb{Z}$. Then, the new code will have a logical dimension of $d_L^\prime = \alpha^{2N}d_L = \lambda^N d_L$. Therefore, it is not generally possible to attain arbitrary logical dimensions through scaling a base lattice.

\begin{corollary}[Gaussian transformation to a square lattice]\label{cor:transform_square_lattice_GKP}
    Suppose a GKP code's underlying symplectic integral lattice, $\Lambda$, has a canonical symplectic Gram matrix of the form $A= \operatorname{diag}\left(d_1, \cdots, d_{N}\right) \otimes \omega$. There exists a Gaussian transformation, $\hat{U}_S$, that transforms code into a GKP code that corresponds to a square lattice, $\Lambda_\text{sq}$, with a generator of the form $M_{\text{sq}} = D_{\text{sq}}\otimes I_2 = \operatorname{diag}\left(\sqrt{d_1}, \cdots, \sqrt{d_{N}}\right) \otimes I_2$.
\end{corollary}
\begin{proof}
    With a Gaussian unitary, $\hat{U}_S$, the quadrature operators transform as 
    \begin{eqnarray}
        \hat{U}_S^\dagger \hat{x}\hat{U}_S = S \hat{x}
    \end{eqnarray}
    where $S$ is a symplectic matrix. Equivalently, one can show that the GKP stabilizer group with a generator matrix $M$ is transformed into another GKP code with $M^\prime = MS^T$. Therefore, from Lemma~\ref{lem:sq_lattice_connect}, it is straightforward to take $S^T = M^{-1}M_{\text{sq}} $, leading to $M^\prime = M_{\text{sq}} = \operatorname{diag}\left(\sqrt{d_1}, \cdots, \sqrt{d_{N}}\right) \otimes I_2$.
\end{proof}

Therefore, we can understand each GKP code as being Gaussian transformed from a square lattice GKP code, where the $i$th mode encodes $d_i$ logical dimensions. For an infinite energy GKP code, the codewords can then be denoted as $\ket{\boldsymbol{\mu}}_0$ where $\boldsymbol{\mu} = \begin{pmatrix}
    \mu_1\\
    \vdots\\
    \mu_N
\end{pmatrix}$ and $\mu_i\in \mathbb{Z}^{d_i}$. The codewords follow as
\begin{eqnarray}
    \ket{\boldsymbol{\mu}}_0 := \hat{U}_S^\dagger \ket{\boldsymbol{\mu}}^{\text{sq}}_0 = \hat{U}_S^\dagger \bigotimes_{i=1}^N \ket{\mu_i}^{\text{sq}}_0
\end{eqnarray}
where the single-mode square lattice codewords are 
\begin{eqnarray}
    \ket{\mu}^{\text{sq}}_0 := \sum_{n\in\mathbb{Z}}\ket{\sqrt{\pi}(dn+\mu)}_x
\end{eqnarray}
with the codewords defined in the conventional Z basis. To obtain its finite-energy counterpart, an overall Gaussian envelope is applied
\begin{align}
    \ket{\boldsymbol{\mu}}_\Delta = N_{\boldsymbol{\mu}} e^{-\Delta^2 \sum_{i=1}^N \hat{n}_i}\ket{\boldsymbol{\mu}}_0
\end{align}
where $\hat{n}_i$ is the number operator for the $i$th mode and $N_{\boldsymbol{\mu}}$ is a normalization factor that is in general dependent on the codeword. We will determine $N_{\boldsymbol{\mu}}$ as a byproduct in Lemma~\ref{lem:single_mode_displaced_code_overlap}. The Gaussian envelope can be decomposed in the displacement operator basis~\cite{PhysRevA.97.032346}. In a single mode, we have
\begin{eqnarray}
    e^{-\Delta^2 \hat{n}} = \frac{1}{\pi\left(1-e^{-\Delta^2}\right)} \int_{-\infty}^\infty d^2 \alpha\hat{D}(\alpha)e^{-\frac{\abs{\alpha}^2}{2\tanh\frac{\Delta^2}{2}}},
\end{eqnarray}
which can be straightforwardly extended to multimode scenarios.

\section{Noise channels\label{sec:appendix_noise_channel}}

In this work, we are mainly concerned with channels such as pure loss, amplification, and thermal loss.

\subsection{Pure loss}

For codes encoded in an oscillator, the excitation loss noise channel, also known as the pure loss channel, has the form of $\mathcal{N}_L(\hat{\rho}) = \sum_{i=0}^\infty \hat{E}_l \hat{\rho} \hat{E}_l^\dagger$, where
\begin{equation}
    \hat{E}_l = \left(\frac{\gamma}{1-\gamma}\right)^{l/2}\frac{\hat{a}^l}{\sqrt{l!}}\left(1-\gamma\right)^{\hat{n}/2},
\end{equation}
and $\hat{n} = \hat{a}^\dagger \hat{a}$ is the number operator. Here, $\gamma$ is the loss parameter and is related to the transmissivity as $\eta = 1-\gamma$. Given the overlap between displaced GKP states, it is ideal if we can expand the Kraus operators of the loss channel in terms of displacement. Such an expression was derived in Ref.~\cite{PhysRevA.97.032346}.

\begin{lemma}[Displacement representation of the loss channel~\cite{PhysRevA.97.032346}] \label{lem:loss_disp_expand}
The loss Kraus operators can be expanded as
    \begin{equation}
    \hat{E}_{l}^{\dagger} \hat{E}_{l^{\prime}}=\int \frac{d^2 \alpha}{\pi} e^{-\frac{1}{2}(1-\gamma)|\alpha|^2}\left\langle l\left|\hat{D}_{\alpha^{\ast}}\right| l^{\prime}\right\rangle \hat{D}_{\alpha \sqrt{\gamma}}.
\end{equation}
\end{lemma}
\begin{proof}
    The Kraus operator combinations
    \begin{eqnarray}
        \hat{E}_{l}^\dagger \hat{E}_k &=& \frac{\sbkt{\frac{\gamma}{1 - \gamma}}^{\frac{l+k}{2}}}{\sqrt{k! l!}}\sbkt{1-\gamma}^{\hat{n}/2} \hat{a}^{\dagger k} \hat{a}^l (1-\gamma)^{\hat{n}/2}\\
        \Tr{\hat{D}_\alpha^\dagger \hat{E}_{l}^\dagger \hat{E}_k} &=& \sum_n\bra{n}\hat{D}_\alpha^\dagger\frac{\sbkt{\frac{\gamma}{1 - \gamma}}^{\frac{l+k}{2}}}{\sqrt{k! l!}}\sbkt{1-\gamma}^{\hat{n}/2} \hat{a}^{\dagger k} \hat{a}^l (1-\gamma)^{\hat{n}/2}\ket{n}
    \end{eqnarray}
    Redefining $p = n-l$, we can expand the displacement operators in the Fock basis and arrive at
\begin{eqnarray}
        \Tr{\hat{D}_\alpha^\dagger \hat{E}_{l}^\dagger \hat{E}_k} &=& \frac{\sbkt{\frac{\gamma}{1 - \gamma}}^{\frac{l+k}{2}}}{\sqrt{k! l!}}\sum_p\bra{p+l}\hat{D}_\alpha^\dagger \ket{p+k}\sbkt{1 - \gamma}^{p + \frac{k + l}{2}}\frac{\sqrt{\sbkt{p+l}!\sbkt{p+k}!}}{p!}\\
        &=& \frac{\sbkt{\frac{\gamma}{1 - \gamma}}^{\frac{l+k}{2}}}{\sqrt{k! l!}}\sbkt{1 - \gamma}^{\frac{k + l}{2}}e^{-\frac{\abs{\alpha}^2}{2}}\sbkt{-\alpha}^{l-k}\sum_{p}\frac{\sbkt{p+k}!}{p!}L_{p+k}^{l-k}\sbkt{\abs{\alpha}^2}\sbkt{1-\gamma}^p\\
        &=& \frac{1}{\gamma}\gamma^{\frac{k-l}{2}}\sqrt{\frac{k!}{l!}}\sbkt{-\alpha}^{l-k}e^{-\frac{1 - \gamma}{\gamma}\abs{\alpha}^2}e^{-\frac{\abs{\alpha}^2}{2}}L_{k}^{l-k}\sbkt{\frac{\abs{\alpha}^2}{\gamma}}\\
        &=& \frac{1}{\gamma}\sandwich{l}{\hat{D}_{\alpha^\ast/\sqrt{\gamma}}}{k}e^{-\frac{\abs{\alpha}^2}{2}\frac{1 - \gamma}{\gamma}}
\end{eqnarray}
Therefore, the decomposition is given by
\begin{eqnarray}
    \hat{E}_{l}^\dagger \hat{E}_k &=& \frac{1}{\gamma}\int d^2 \alpha \frac{e^{-\frac{1}{2}\frac{1-\gamma}{\gamma}\abs{\alpha}^2}}{\pi}\sandwich{l}{\hat{D}_{\alpha^\ast/\sqrt{\gamma}}}{k}\hat{D}_\alpha\\
        &=& \int d^2 \alpha \frac{e^{-\frac{1}{2}\sbkt{1-\gamma}\abs{\alpha}^2}}{\pi}\sandwich{l}{\hat{D}(\alpha^\ast)}{k}\hat{D}_{\alpha\sqrt{\gamma}}
\end{eqnarray}

\end{proof}

\subsection{Amplification channel}

For the amplification channel, we have its Kraus operator form as $\mathcal{N}_A(\hat{\rho}) = \sum_{i=0}^\infty \hat{A}_l \hat{\rho} \hat{A}_l^\dagger$
\begin{eqnarray}
    A_l &:=& \sqrt{\frac{1}{l!}\frac{(G-1)^l}{G}}G^{-\frac{n}{2}}a^{\dagger l}.
\end{eqnarray}
Given the similar forms between the amplification and the loss channel, we can expect a similar decomposition in the displacement representation.

\begin{lemma}[Displacement representation of the amplification channel] \label{lem:amp_disp_expand}
    The amplification Kraus operators can be expanded as
        \begin{equation}
        \hat{A}_{l}^{\dagger} \hat{A}_{l^{\prime}}=\int \frac{d^2 \alpha}{\pi} e^{-\frac{1}{2}G|\alpha|^2}\left\langle l\left|\hat{D}_{\alpha^{\ast}}\right| l^{\prime}\right\rangle \hat{D}_{\alpha \sqrt{G-1}}.
    \end{equation}
    \end{lemma}
    \begin{proof}
        The proof is similar to the proof of Lemma~\ref{lem:loss_disp_expand} and is omitted for simplicity.
    \end{proof}

\section{QEC matrix and near-optimal fidelity}

One key concept we use in this work is the near-optimal fidelity based on the QEC matrix, recently proposed by Ref.~\cite{zheng2024near}. For completeness, we review the key ideas here.

\begin{definition}[QEC matrix]\label{def:qec_mat} Consider a $d_L$-dimensional QEC code defined through codewords $\set{\ket{\mu_L}}$ and a noise channel admitting a Kraus form of $\set{\hat{N}_l}$ with Kruas order $L$. The QEC matrix is a $d_L L\times d_L L$ matrix, defined as 
    \begin{eqnarray}
        M_{[\mu l],[\nu k]} = \bra{\mu_L}\hat{N}_l^\dagger \hat{N}_{k}\ket{\nu_L}.
    \end{eqnarray}
\end{definition}

The QEC matrix is proposed first as a tool to check if the encoding allows for exact recovery of the noise, also known as the Knill-Laflamme conditions. In fact, the QEC matrix contains more information than a Yes/No answer to the question of whether a code is an exact code. For approximate codes, we can discover their optimal fidelity by optimizing the recovery channel 
\begin{align}
    F^{\text{opt}} := \max_{\mathcal{R}} F\left(\mathcal{R}\circ\mathcal{N}\circ\mathcal{E}\right) = F\left(\mathcal{R}^{\text{opt}}\circ\mathcal{N}\circ\mathcal{E}\right).
\end{align}
The fidelity metric of choice is the channel fidelity, defined as 
\begin{eqnarray}
    F\left(\mathcal{Q}\right) := \bra{\Phi} \mathcal{Q} \otimes\mathcal{I}_R\left(\ket{\Phi}\bra{\Phi}\right)\ket{\Phi},
\end{eqnarray}
where $\ket{\Phi}$ is the purified maximally mixed state, $\mathcal{Q}$ is an arbitrary quantum channel, and $\mathcal{I}_R$ is the identity channel acting on the reference ancillary system. With such a metric, the optimal fidelity can be found through convex optimization. The near-optimal fidelity is an optimization-free benchmark that only depends on the QEC matrix, as defined below.

\begin{lemma}[The near-optimal fidelity; Theorem 1, Ref.~\cite{zheng2024near}]\label{lem:near_optimal}
    For a $d_L$-dimensional encoding, $\mathcal{E}$, and a noise channel, $\mathcal{N}$, the near-optimal fidelity is
    \begin{align}
        \tilde{F}^{\text{opt}}=\frac{1}{d_L^2}\left\|\operatorname{Tr}_L \sqrt{M}\right\|_F^2,
    \end{align}
    where $M$ is the QEC matrix, $\left(\operatorname{Tr}_L B\right)_{l,k} = \sum_\mu B_{[\mu l], [\mu k]}$ denotes the partial trace over the code space indices, and $||\cdot||_F$ is the Frobenius norm. The near-optimal fidelity gives a two-sided bound on the optimal fidelity as
    \begin{eqnarray}\label{eq:two_sided_bound_SM}
    \frac{1}{2}\left(1 - \tilde{F}^{\text{opt}} \right)\leq 1 - F^{\text{opt}} \leq 1 - \tilde{F}^{\text{opt}}. 
    \end{eqnarray}
\end{lemma}

While the exact expression of the near-optimal fidelity gives a computationally efficient approach to finding the near-optimal performance of any QEC codes, it is not very useful analytically. Therefore, the matrix square root can be expanded perturbatively.
\begin{lemma}[Perturbative form of the near-optimal fidelity; Corollary 1, Ref.~\cite{zheng2024near}]\label{lem:near_optimal_pert}
    The noise channel's Kraus representation can be chosen such that $\frac{1}{d_L}\operatorname{Tr}_L M =D$, with $M$ being the QEC matrix and $D$ being a diagonal matrix. With the residual matrix $\Delta M := M - I_L \otimes D$, the near-optimal infidelity has a perturbative expansion through
   \begin{eqnarray}
   1 - \tilde{F}^{\text{opt}} = \frac{1}{d_L}\norm{f(D)\odot \Delta M }_F^2 + \mathcal{O}\left(\frac{1}{d_L}\norm{f(D)\odot \Delta M }_F^3\right)
   \end{eqnarray}
   where $f(D) _{[\mu l],[\nu k]} = \frac{1}{\sqrt{D_{ll}} + \sqrt{D_{kk}}}$ and the Hadamard product $\left(A\odot B\right)_{ij} = A_{ij}B_{ij}$.
\end{lemma}
Based on Lemma~\ref{lem:near_optimal_pert}, we have the perturbative form of the near-optimal infidelity as 
\begin{eqnarray}
    \tilde{\epsilon} := \frac{1}{d_L}\norm{f(D)\odot \Delta M }_F^2,
\end{eqnarray}
such that 
\begin{eqnarray}
    1-\tilde{F}^{\text{opt}} = \tilde{\epsilon} + O\left(\tilde{\epsilon}^{3/2}\right).\label{eq:pert_expression_appendix}
\end{eqnarray}
Here the residual term is bounded under the assumption that $d_L$ is a finite constant. The perturbative form is quite useful for obtaining analytical expressions for qubit and bosonic codes, as shown in Ref.~\cite{zheng2024near}. Lemma~\ref{lem:near_optimal_pert} provides one possible definition of $D$. Nevertheless, under such a definition, it could be cumbersome to compute the analytical form of $D$ since it requires diagonalization. Alternatively, we can relax the definition and instead define $D$ to be a truncated diagonal matrix. The residual part is still defined as $\Delta M := M - I_L \otimes D$. As a result, $\Tr D$ does not necessarily equal 1 as before. Qualitatively, to obtain an analytical expression, we sacrifice some of the correctable parts by dividing them into the residual matrix. 

In Eq.~\eqref{eq:pert_expression_appendix} we assumed $d_L$ to be finite to bound the residual term from the approximation. Such an assumption does not hold when we are concerned with information-theoretic properties such as achievable rates. Nevertheless, it has been shown in Ref.~\cite{zheng2024near} that even with increasing logical dimension, a vanishing first-order approximation is sufficient to guarantee vanishing near-optimal infidelity, which is stated formally below.
\begin{lemma}[Corollary 3, Ref.~\cite{zheng2024near}]\label{lem:vanish_pert_form}
    As $\tilde{\epsilon} \to 0$, the near-optimal fidelity
  \begin{align}
       \tilde{F}^{\text{opt}} \to \Tr{D}.\label{eq:vanishing_pert_form}
  \end{align}
\end{lemma}

The above tools allow us to develop the key corollary we will use to obtain the achievable rate.

\begin{corollary}\label{cor:achieve_rate_near_optimal_fid}
    For a quantum code encoding a $d_L$-dimensional logical space and passing through $N$ independent uses of a channel $\mathcal{Q}^{\otimes N}$, the rate 
    \begin{eqnarray}
            R = \frac{1}{N}\log_2 d_L 
    \end{eqnarray}
    is achievable if $\tilde{\epsilon} \to 0$ and $D$ is chosen such that $\Tr D \to 1$.
\end{corollary}
\begin{proof}
    This corollary is a direct consequence of Lemma~\ref{lem:vanish_pert_form}, the two-sided bounds presented in Eq.~\eqref{eq:two_sided_bound_SM}, and the definition of achievable rates.
\end{proof}

\section{QEC matrix of GKP}
\subsection{Pure loss}
We first focus on the example of GKP under pure loss. The high-level idea is to view multi-mode GKP as Gaussian-transformed multi-mode square lattice GKP, which has a tensor product structure. Therefore, many derivations for single-mode square lattice GKP can be straightforwardly applied in multimode general lattices. Therefore, we start with some properties of square-lattice GKP in single mode.

\subsubsection{Single mode \label{ap:proveM}}
For simplicity, we begin with single-mode GKP and later generalize to multimode scenarios. Since GKP is characterized by its underlying lattice, it is relatively simple to work with displacement operators, which essentially apply a translation on the GKP lattice. In fact, we can derive the following lemma for square lattice GKP code.

\begin{lemma}[Overlap of displaced single-mode square lattice GKP code; \cite{PhysRevA.97.032346}, Eq.~(D11)] \label{lem:single_mode_displaced_code_overlap}
    For a single mode square lattice GKP code encoding $d-$dimensional logical space, the overlap between finite energy displaced codewords is given by
    \begin{align}
    {}_\Delta\bra{\mu_{1}} \hat{D}(\alpha)\ket{\mu_{ 2}}_\Delta  =& \frac{N_{\mu}N_{\nu}}{2\sqrt{\pi}(1-e^{-2\Delta^2})}\sum_{n_1,n_2 \in \mathbb{Z}} e^{i\pi (n_1 + \frac{\mu+\nu}{d})n_2}e^{-\frac{\abs{\sqrt{2\pi}C\bold{L} - \alpha}^2}{4\tanh \frac{\Delta^2}{2}}}e^{-\frac{\tanh \frac{\Delta^2}{2}}{4}\abs{\sqrt{2\pi}C\bold{L} + \alpha}^2}\\
    =& \left(1 +O\left(e^{-\frac{\pi}{d}\left(n_\Delta + \frac{1}{2}\right)}\right)\right)\sum_{n_1,n_2 \in \mathbb{Z}} e^{i\pi (n_1 + \frac{\mu+\nu}{d})n_2}e^{-\frac{\abs{\sqrt{2\pi}C\bold{L} - \alpha}^2}{4\tanh \frac{\Delta^2}{2}}}e^{-\frac{\tanh \frac{\Delta^2}{2}}{4}\abs{\sqrt{2\pi}C\bold{L} + \alpha}^2} \label{eq:singlemode_QEC_squareuDmu}
    \end{align}
    where $D_\alpha$ is the displacement operator, $\bold{L}:= \sqrt{\frac{1}{d}}\left(dn_1+\mu-\nu, n_2\right)^T$, $C$ is the linear map defined in Eq.~\eqref{eq:conversion_T_D}, and $n_\Delta = \frac{1}{e^{2 \Delta^2} -1}$.
\end{lemma}

\begin{proof}
We can focus first on the infinite energy scenario. Expanding the codewords on the position basis,
\begin{eqnarray}
     {}_0\bra{\mu} \hat{D}(\alpha)\ket{\nu}_0 
    &=& \sum_{n_1,n_2 \in \mathbb{Z}}{}_{\hat{x}}\bra{\sqrt{\pi}(dn_1+\mu)}_xD_\alpha \ket{\sqrt{\pi}(dn_2+\nu)}_x\\
    &=& \sqrt{\frac{\pi}{2d}}\sum_{n_1,n_2 \in \mathbb{Z}}e^{i\pi (n_1+\frac{\mu+\nu}{d})n_2} \delta^2 \left(\alpha - \sqrt{\frac{\pi}{d}}(dn_1+\mu-\nu+in_2)\right)\\
    &=& \sqrt{\frac{\pi}{2d}}\sum_{n_1,n_2 \in \mathbb{Z}}e^{i\pi (n_1+\frac{\mu+\nu}{d})n_2} \delta^2 \left(\alpha - \sqrt{2\pi}C\bold{L}\right)
    \label{eq:inf_single_mode_displace_overlap}
\end{eqnarray}
where $\alpha = \alpha_1 + i\alpha_2$ and $\bold{L}:= \frac{1}{\sqrt{{d}}}\left(dn_1+\mu-\nu, n_2\right)^T$.  Here, we have applied the displacement operator decomposition in the position basis, ${}_{\hat{x}}\bra{q_1}D_\alpha \ket{q_2}_{\hat{x}}
    = \frac{1}{\sqrt{2}}e^{i\alpha_1 \alpha_2}e^{i\sqrt{2}\alpha_2 q_2}\delta (\alpha_1 + \frac{q_2-q_1}{\sqrt{2}})$, and the Poisson summation. The finite energy matrix elements can be written as an integral of infinite energy matrix elements,
\begin{align}
    {}_\Delta\bra{\mu} \hat{D}(\alpha)\ket{\nu}_\Delta=& \frac{N_{\mu}N_{\nu}}{\pi^2(1-e^{-\Delta^2})^2}\int d^2 \beta \int d^2 \gamma e^{-\frac{|\beta|^2+|\gamma|^2}{2 \tanh\frac{\Delta^2}{2} }} {}_0\bra{\mu} \hat{D}(\beta)\hat{D}(\alpha)\hat{D}(\gamma)\ket{\nu}_0 \\
    =& \frac{N_{\mu}N_{\nu}}{\pi^2(1-e^{-\Delta^2})^2}\int d^2 \beta \int d^2 \gamma e^{-\frac{|\beta|^2+|\gamma|^2}{2 \tanh\frac{\Delta^2}{2} }} e^{i\Im{\beta\alpha^*+\beta\gamma^*+\alpha\gamma^*}} {}_0\bra{\mu} \hat{D}(\alpha+\beta+\gamma)\ket{\nu}_0, \label{eq:fi_mDn}\\
    =& \frac{N_{\mu}N_{\nu}}{2\sqrt{\pi}(1-e^{-2\Delta^2})}\sum_{n_1,n_2 \in \mathbb{Z}} e^{i\pi (n_1 + \frac{\mu+\nu}{d})n_2}e^{-\frac{\abs{\sqrt{2\pi}C\bold{L} - \alpha}^2}{4\tanh \frac{\Delta^2}{2}}}e^{-\frac{\tanh \frac{\Delta^2}{2}}{4}\abs{\sqrt{2\pi}C\bold{L} + \alpha}^2}\\
    =& \left(1 +O\left(e^{-\frac{\pi}{d}\left(n_\Delta + \frac{1}{2}\right)}\right)\right)\sum_{n_1,n_2 \in \mathbb{Z}} e^{i\pi (n_1 + \frac{\mu+\nu}{d})n_2}e^{-\frac{\abs{\sqrt{2\pi}C\bold{L} - \alpha}^2}{4\tanh \frac{\Delta^2}{2}}}e^{-\frac{\tanh \frac{\Delta^2}{2}}{4}\abs{\sqrt{2\pi}C\bold{L} + \alpha}^2} 
\end{align}
where we make use of the composition rule of the displacement operator. The last equality is obtained by considering $\mu= \nu$ and $\alpha = 0$ and by the definition of normalization factors, $N_\mu$, we have 
\begin{eqnarray}
    {}_\Delta\bra{\mu} \ket{\mu}_\Delta &=& \frac{N^2_{\mu}}{2\sqrt{\pi}(1-e^{-2\Delta^2})}\sum_{n_1,n_2 \in \mathbb{Z}} e^{i\pi (n_1 + \frac{\mu+\nu}{d})n_2}e^{-\pi (n_\Delta + \frac{1}{2})\abs{\bold{L}}^2} = 1.
\end{eqnarray}
Here we have defined $n_\Delta = \frac{1}{e^{2 \Delta^2} -1}$ and used the identity $\sqrt{2}\abs{CL} = \abs{L}$. Physically, $n_\Delta$ is the average photon number of the GKP code, which will become clear in Corollary~\ref{cor:average_energy_single_mode}. The leading order term is when $n_{1,2} = 0$, and the next leading order term ($n_1 = 0, n_2 = \pm 1$) will be exponentially suppressed by the photon number. In the limit of large energy, we have 
\begin{eqnarray}
    N_{\mu}^2 = 2\sqrt{\pi}(1-e^{-2\Delta^2}) \left(1 + O\left(e^{-\frac{\pi}{d}\left(n_\Delta + \frac{1}{2}\right)}\right)\right).\label{eq:normalization}
\end{eqnarray}
Here, the terms that are dependent on $\mu$ are absorbed in the second term.
\end{proof}

As a consequence of Lemma~\ref{lem:single_mode_displaced_code_overlap} and Lemma~\ref{lem:loss_disp_expand}, we arrive at the analytic form of the QEC matrix. While a similar form has been derived in Eq.~(D16) of Ref.~\cite{PhysRevA.97.032346}, the form we provide here is exact, without approximations based on the large energy regime assumption.

\begin{lemma}[QEC matrix of the single-mode square lattice GKP code]\label{lem:single_mode_QEC}
    For a single mode square lattice GKP code encoding $d-$dimensional logical space, the QEC matrix of the GKP code under pure loss has the form of
\begin{align}
    M_{[\mu l],[\nu k]} &= \left(1 +O\left(e^{-\frac{\pi}{d}\left(n_\Delta + \frac{1}{2}\right)}\right)\right)\frac{t^{\frac{l+k}{2}}}{\gamma n_\Delta + 1}\sum_{n_1,n_2 \in \mathbb{Z}}e^{i\pi (n_1 + \frac{\mu+\nu}{d})n_2}e^{-\frac{\pi}{2}\frac{(1-\gamma)}{ (\gamma+\frac{1}{n_{\Delta}})}\left|\bold{L}\right|^2}\left\langle l\left|\hat{D}\left({\sqrt{\frac{n_\Delta + 1}{\gamma n_\Delta + 1}}\left(\sqrt{2\pi}C\bold{L}\right)^\ast}\right) \right| k\right\rangle
    \label{eq:singlemode_QEC_square}
\end{align}
where the average energy $n_{\Delta} :=\frac{1}{e^{2\Delta^2}-1}$, the thermal factor $t = \frac{\gamma n_\Delta}{\gamma n_\Delta + 1}$, and $\bold{L}:= \frac{1}{\sqrt{d}}\left(dn_1+\mu-\nu, n_2\right)^T$.
\end{lemma}
\begin{proof}
From the definition of the QEC matrix and invoking Lemma~\ref{lem:single_mode_displaced_code_overlap} and Lemma~\ref{lem:loss_disp_expand}, we have
\begin{align}
    & M_{[\mu l],[\nu k]} \\
    :=& _{\Delta}\bra{\mu}\hat{E}_l^\dagger\hat{E}_k\ket{\nu}_{\Delta}\\
    =&\left(1 +O\left(e^{-\frac{\pi}{d}\left(n_\Delta + \frac{1}{2}\right)}\right)\right)\int \frac{d^2 \alpha}{\pi} e^{-\frac{1-\gamma}{2}\abs{\alpha}^2}\bra{l}\hat{D}(\alpha^\ast)\ket{k}\sum_{n_1,n_2 \in \mathbb{Z}} e^{i\pi (n_1 + \frac{\mu+\nu}{d})n_2} e^{-\frac{\abs{\sqrt{2\pi}C\bold{L} - \sqrt{\gamma}\alpha}^2}{4t_\Delta}}e^{-\frac{t_\Delta}{4}\abs{\sqrt{2\pi}C\bold{L} + \sqrt{\gamma}\alpha}^2},
\end{align}
where for simplicity we denote $t_\Delta = \tanh \frac{\Delta^2}{2}$, $\bold{L}:= \frac{1}{\sqrt{d}}\left(dn_1+\mu-\nu, n_2\right)^T$, and $C\bold{L} := |C\bold{L}|e^{i\theta_L}\in \mathbb{C}$. Notice that the displacement operator can be expanded in the Fock basis in terms of the generalized Laguerre function:
\begin{align}
    \bra{l}\hat{D}(\alpha)\ket{k}=e^{-\frac{|\alpha|^2}{2}} \sqrt{\frac{k !}{l !}} L_{k}^{l-k}\left(|\alpha|^2\right) \alpha^{l-k}.
\end{align}
for $l\geq k$. When $l\leq k$, it can be computed through $\left(\bra{k}\hat{D}(-\alpha)\ket{l}\right)^\ast$. Rewriting the integral in polar coordinates, we have
\begin{align}
    M_{[\mu l],[\nu k]} = &\left(1 +O\left(e^{-\frac{\pi}{d}\left(n_\Delta + \frac{1}{2}\right)}\right)\right) \sqrt{\frac{k !}{l !}}\sum_{n_1,n_2 \in \mathbb{Z}} e^{i\pi (n_1 + \frac{\mu+\nu}{d})n_2} e^{-(n_\Delta+\frac{1}{2})|\sqrt{2\pi}C\bold{L}|^2}\\&\int_{0}^{\infty} |\alpha| d|\alpha| e^{-(1-\frac{\gamma}{2}+\gamma (n_\Delta+\frac{1}{2}))|\alpha|^2}L_{k}^{l-k}(|\alpha|^2)|\alpha|^{l-k}\int_{0}^{2\pi}\frac{d\theta}{\pi}e^{-2\sqrt{\gamma n_\Delta\left(n_\Delta+1\right)}|\sqrt{2\pi}C\bold{L}| |\alpha|\cos (\theta-\theta_L)}e^{-i(l-k)\theta},
\end{align}
where $\alpha = \abs{\alpha}e^{i\theta}$ and $n_{\Delta} :=\frac{1}{e^{2\Delta^2}-1}$. We then apply two mathematical identities~\cite{PhysRevA.97.032346}
\begin{align}
    I_n(z)&=\frac{1}{\pi} \int_0^\pi e^{z \cos \theta} \cos (n \theta) d \theta,\\
    \int_0^{\infty} d x x^{\frac{\lambda}{2}} e^{-p x} I_\lambda(2 b \sqrt{x}) L_n^{(\lambda)}(x)& =b^\lambda \frac{(p-1)^n}{p^{\lambda+n+1}} e^{\frac{b^2}{p}} L_n^{(\lambda)}\left(\frac{b^2}{p(p-1)}\right).
\end{align}
With these simplifications, we reach the final expression
\begin{align}
      & M_{[\mu l],[\nu k]}\\
      =&\left(1 +O\left(e^{-\frac{\pi}{d}\left(n_\Delta + \frac{1}{2}\right)}\right)\right) \sqrt{\frac{k !}{l !}}\sum_{n_1,n_2 \in \mathbb{Z}} e^{i\pi (n_1 + \frac{\mu+\nu}{d})n_2} e^{-(n_\Delta+\frac{1}{2})|\sqrt{2\pi}C\bold{L}|^2}e^{-i(l-k)\theta}\left(\sqrt{\gamma n_\Delta (n_\Delta+1)}|\sqrt{2\pi}C\bold{L}|\right)^{l-k}\\
     &\frac{(-\frac{\gamma}{2}+\gamma (n_\Delta+\frac{1}{2}))^k}{(1-\frac{\gamma}{2}+\gamma (n_\Delta+\frac{1}{2}))^{l+1}}e^{|\sqrt{2\pi}C\bold{L}|^2\frac{n_\Delta+1}{\gamma n_\Delta+1}} L_{k}^{l-k}\left(|\sqrt{2\pi}C\bold{L}|^2\frac{n_\Delta+1}{\gamma n_\Delta+1}\right)\\
     =&\left(1 +O\left(e^{-\frac{\pi}{d}\left(n_\Delta + \frac{1}{2}\right)}\right)\right)\frac{t^{\frac{l+k}{2}}}{\gamma n_\Delta + 1}\sum_{n_1,n_2 \in \mathbb{Z}}e^{i\pi (n_1 + \frac{\mu+\nu}{d})n_2}e^{-\frac{\pi}{2}\frac{(1-\gamma)}{ (\gamma+\frac{1}{n_{\Delta}})}\left|\bold{L}\right|^2}\left\langle l\left|\hat{D}\left({\sqrt{\frac{n_\Delta + 1}{\gamma n_\Delta + 1}}\left(\sqrt{2\pi}C\bold{L}\right)^\ast}\right) \right| k\right\rangle.
\end{align}
\end{proof}

One quantity that we have defined is $n_\Delta$, which is considered as the average energy. Indeed, in the large energy regime, it converges to the often-quoted energy of $\frac{1}{2\Delta^2} - \frac{1}{2}$. Below, we show more rigorously this is a more precise expression for the average energy.

\begin{corollary}[Average energy of single-mode square lattice GKP code] \label{cor:average_energy_single_mode}
For single-mode square lattice GKP codes, the average energy in the codeword,
\begin{eqnarray}
    \bar{n}:=\frac{1}{d}\Tr{\hat{P}_L\hat{n}} = n_\Delta + O\left(e^{-\frac{\pi}{2d}\left(n_\Delta + \frac{1}{2}\right)}\right)
\end{eqnarray}
where $n_{\Delta}:=\frac{1}{e^{2 \Delta^2}-1}$.
\end{corollary}
\begin{proof}
    With the no-jump Kraus operator of the pure loss channel, we have that $\hat{E}_0^\dagger \hat{E}_0 = (1-\gamma)^{\hat{n}}$. Therefore, 
    \begin{eqnarray}
        \bar{n}&:=&\frac{1}{d}\Tr{\hat{P}_L\hat{n}} \\
        &=&-\frac{1}{d}\sum_{\mu=0}^{d-1}\frac{d}{d\gamma} M_{[\mu 0], [\mu 0]}\Bigg|_{\gamma = 0}\\
        &=&-\frac{1}{d}\sum_{\mu=0}^{d-1}\frac{d}{d\gamma} \frac{N^2_{\mu}}{2\sqrt{\pi}(1-e^{-2\Delta^2})}\frac{1}{\gamma n_\Delta + 1}\sum_{n_1,n_2 \in \mathbb{Z}}e^{i\pi (n_1 + \frac{\mu+\nu}{d})n_2}e^{-\frac{\pi}{2}\frac{2-\gamma + \frac{1}{n_\Delta}}{ \gamma+\frac{1}{n_{\Delta}}}\left|\bold{L}\right|^2}\Bigg|_{\gamma = 0}\\
        &=& n_\Delta + O\left(e^{-\frac{\pi}{d}\left(n_\Delta + \frac{1}{2}\right)}\right)
    \end{eqnarray}
    where we have adopted the relations given in Eq.~\eqref{eq:normalization}.
\end{proof}

It is worth noting that while we only derived the energy for single-mode square lattices, it can be easily generalized to multimode and general lattices. The modifications would be that $\bar{n} = Nn_\Delta$, where $N$ is the number of modes. Moreover, the error term would be $O\left(e^{-\pi\left(n_\Delta + \frac{1}{2}\right)\abs{x}_{\text{min}}^2}\right)$, where $\abs{x}_{\text{min}}$ is the Euclidean distance of the shortest lattice vector of the symplectic dual lattice of the corresponding GKP lattice. This fact will become clear soon.

\subsubsection{Multi-mode}\label{ap:multi-mode}

Many of the calculations can be inherited from the single-mode case.

\begin{lemma}[Overlap of displaced GKP states]\label{th:mul_2}
    Suppose an infinite-energy N-mode GKP code corresponds to an underlying lattice of $\Lambda$, whose canonical symplectic Gram matrix $A= \operatorname{diag}\left(\bold{d}\right) \otimes \omega$. The overlap between displaced GKP states has the form of
    \begin{align}
         _\Delta\bra{\boldsymbol{\mu}}\hat{D}(\boldsymbol{\alpha})\ket{\boldsymbol{\nu}}_\Delta = (\frac{1}{2\sqrt{\pi}(1-e^{-2\Delta^2})})^N \sum_{\bold{n}_1,\bold{n}_2\in \mathbb{Z}^N} e^{i\pi  \bold{n}^T_{2}\left(\bold{n}_{1}+\left(\boldsymbol{\mu}+\boldsymbol{\nu}\right)\oslash \bold{d}\right)}
         e^{-\frac{\abs{\sqrt{2\pi}C\boldsymbol{L} - \boldsymbol{\alpha}}^2}{4\tanh \frac{\Delta^2}{2}}}e^{-\frac{\tanh \frac{\Delta^2}{2}}{4}\abs{\sqrt{2\pi}C\boldsymbol{L} + \boldsymbol{\alpha}}^2}.
    \end{align}
    where the Hadamard division is defined as $\left(\bold{a}\oslash \bold{b}\right)_i = a_i/b_i$ and
    \begin{eqnarray}
        \bold{L} &:=&  S^{-1}\left(D^{-1/2}_{\text{sq}} \otimes I_2\right)\left(d_1 n_{1,1}+\mu_1-\nu_1, n_{2,1}, \dots, d_N n_{1,N}+\mu_N-\nu_N, n_{2,N}\right)^T
    \end{eqnarray}
    and $D_{\text{sq}} := \text{diag}\left(d_1, \dots, d_N\right)$.
\end{lemma}
\begin{proof}
    Recall that a GKP state can be Gaussian transformed into a square lattice GKP through
    \begin{eqnarray}
        \ket{\boldsymbol{\mu}}_0 = \hat{U}^\dagger_S\ket{\boldsymbol{\mu}}^\text{sq}_0
    \end{eqnarray}
    where $\hat{U}_S$ corresponds to the symplectic transform given in Corrolary~\ref{cor:transform_square_lattice_GKP}. $\ket{\boldsymbol{\mu}}_0$ and $\ket{\boldsymbol{\mu}}^\text{sq}_0$ contains the same logical information, $\boldsymbol{\mu}$, but lives on different lattices. The underlying lattice of the Gaussian transformed GKP codewords $\ket{\boldsymbol{\mu}}^\text{sq}_0$ is generated by $M_{\text{sq}} = \operatorname{diag}\left(\sqrt{d_1}, \cdots, \sqrt{d_{N}}\right) \otimes I_2$. Since $\hat{U}_S^\dagger \hat{x} \hat{U}_S = S\hat{x}$, we have $\hat{U}_S \hat{T}(\bold{u}) \hat{U}^\dagger_S = \hat{T}(S\bold{u})$. We can then transform as
    \begin{eqnarray}
        & &_0\bra{\boldsymbol{\mu}}\hat{D}(\boldsymbol{\alpha})\ket{\boldsymbol{\nu}}_0\\
        &=& _0^\text{sq}\bra{\boldsymbol{\mu}} U_S \hat{D}(\boldsymbol{\alpha}) U^\dagger_S \ket{\boldsymbol{\nu}}^\text{sq}_0 \\
        &=& _0^\text{sq}\bra{\boldsymbol{\mu}} \hat{D}(CSC^{-1}\boldsymbol{\alpha}) \ket{\boldsymbol{\nu}}^\text{sq}_0 \\
        &=& \prod_{i=1}^N {}_0^\text{sq}\bra{\mu_i} \left(\hat{D}(CSC^{-1}\boldsymbol{\alpha})\right)_i \ket{\nu_i}^\text{sq}_0 \\
        &=& \frac{\sqrt{\pi}}{2}\sum_{\bold{n}_1,\bold{n}_2\in \mathbb{Z}^N} e^{i\pi  \bold{n}^T_{2}\left(\bold{n}_{1}+\left(\boldsymbol{\mu}+\boldsymbol{\nu}\right)\oslash \bold{d}\right)} \delta^{2N} (CSC^{-1} \boldsymbol{\alpha} - \sqrt{2\pi}C\bold{L}_{\text{sq}})\\
        &=& \frac{\sqrt{\pi}}{2 |\det S |}\sum_{\bold{n}_1,\bold{n}_2\in \mathbb{Z}^N} e^{i\pi  \bold{n}^T_{2}\left(\bold{n}_{1}+\left(\boldsymbol{\mu}+\boldsymbol{\nu}\right)\oslash \bold{d}\right)} \delta^{2N} (\boldsymbol{\alpha}   - \sqrt{2\pi}C S^{-1}\bold{L}_{\text{sq}})\\
        &=& \frac{\sqrt{\pi}}{2 }\sum_{\bold{n}_1,\bold{n}_2\in \mathbb{Z}^N} e^{i\pi  \bold{n}^T_{2}\left(\bold{n}_{1}+\left(\boldsymbol{\mu}+\boldsymbol{\nu}\right)\oslash \bold{d}\right)}\delta^{2N} (\boldsymbol{\alpha}   - \sqrt{2\pi}C\bold{L})
    \end{eqnarray}
    where $\bold{L} := S^{-1}\bold{L}_{\text{sq}}$ with $\bold{L}_{\text{sq}}:= \left(D_{\text{sq}}^{-1/2} \otimes I_2\right)\left(d_1 n_{1,1}+\mu_1-\nu_1, n_{2,1}, \dots, d_N n_{1,N}+\mu_N-\nu_N, n_{2,N}\right)^T$ and $D_{\text{sq}} := \text{diag}\left(d_1, \dots, d_N\right)$. With an abuse of notation, we write $\left(\hat{D}(S\boldsymbol{\alpha})\right)_i$ to represent the displacement component in the $i$th mode. In the derivation, we have applied the result for single mode square lattice given in Eq.~\eqref{eq:inf_single_mode_displace_overlap}, and that $S$ is a symplectic matrix with $\abs{\det S} = 1$. With the infinite energy expression, it can be similarly extended to finite energy conditions
     \begin{align}
       & _\Delta\bra{\boldsymbol{\mu}}\hat{D}(\boldsymbol{\alpha})\ket{\boldsymbol{\nu}}_\Delta \\
       =&N_{\boldsymbol{\mu}}N_{\boldsymbol{\nu}}{}_0\bra{\boldsymbol{\mu}}e^{-\Delta^2\sum_{i=1}^N \hat{n}_i}\hat{D}(\boldsymbol{\alpha})e^{-\Delta^2\sum_{i=1}^N \hat{n}_i}\ket{\boldsymbol{\nu}}_0\\
       =& \frac{N_{\boldsymbol{\mu}}N_{\boldsymbol{\nu}}}{\pi^{2N}\left(1-e^{-\Delta^2}\right)^{2N}} \int d^{2N} \boldsymbol{\beta} \int d^{2N} \boldsymbol{\gamma} e^{-\frac{|\boldsymbol{\beta}|^2+|\boldsymbol{\gamma}|^2}{2 t}} {}_0\bra{\boldsymbol{\mu}}\hat{D}(\boldsymbol{\beta})\hat{D}(\boldsymbol{\alpha})\hat{D}(\boldsymbol{\gamma})\ket{\boldsymbol{\nu}}_0 \\
       =& \frac{N_{\boldsymbol{\mu}}N_{\boldsymbol{\nu}}}{\pi^{2N}\left(1-e^{-\Delta^2}\right)^{2N}} \int d^{2N} \boldsymbol{\beta} \int d^{2N} \boldsymbol{\gamma} e^{-\frac{|\boldsymbol{\beta}|^2+|\boldsymbol{\gamma}|^2}{2 t}} e^{i \Im \{ \boldsymbol{\beta}^T \boldsymbol{\alpha}^\ast  + \boldsymbol{\beta}^T \boldsymbol{\gamma}^\ast + \boldsymbol{\alpha}^T \boldsymbol{\gamma}^\ast\}} {}_0 \bra{\boldsymbol{\mu}}\hat{D}(\boldsymbol{\beta} + \boldsymbol{\alpha} +\boldsymbol{\gamma})\ket{\boldsymbol{\nu}}_0 \\
       =& N_{\boldsymbol{\mu}}N_{\boldsymbol{\nu}}(\frac{1}{2 \pi^{3 / 2}\left(1-e^{-\Delta^2}\right)^2})^N \int d^{2N} \boldsymbol{\beta} \int d^{2N} \boldsymbol{\gamma} e^{-\frac{|\boldsymbol{\beta}|^2+|\boldsymbol{\gamma}|^2}{2 t}} e^{i \Im \{ \boldsymbol{\beta}^T \boldsymbol{\alpha}^\ast  + \boldsymbol{\beta}^T \boldsymbol{\gamma}^\ast + \boldsymbol{\alpha}^T \boldsymbol{\gamma}^\ast\}} \\
       & \sum_{\bold{n}_1,\bold{n}_2\in \mathbb{Z}^N} e^{i\pi  \bold{n}^T_{2}\left(\bold{n}_{1}+\left(\boldsymbol{\mu}+\boldsymbol{\nu}\right)\oslash \bold{d}\right)}\delta^{2N} (\boldsymbol{\beta} + \boldsymbol{\alpha} +\boldsymbol{\gamma}   - C\bold{L})\\
       =& N_{\boldsymbol{\mu}}N_{\boldsymbol{\nu}}(\frac{1}{2\sqrt{\pi}(1-e^{-2\Delta^2})})^N \sum_{\bold{n}_1,\bold{n}_2\in \mathbb{Z}^N} e^{i\pi  \bold{n}^T_{2}\left(\bold{n}_{1}+\left(\boldsymbol{\mu}+\boldsymbol{\nu}\right)\oslash \bold{d}\right)}
       e^{-\frac{\abs{\sqrt{2\pi}C\boldsymbol{L} - \boldsymbol{\alpha}}^2}{4\tanh \frac{\Delta^2}{2}}}e^{-\frac{\tanh \frac{\Delta^2}{2}}{4}\abs{\sqrt{2\pi}C\boldsymbol{L} + \boldsymbol{\alpha}}^2}\\
       =& \left(1 +O\left(e^{-\pi\left(n_\Delta + \frac{1}{2}\right)\abs{x}_{\text{min}}^2}\right)\right) \sum_{\bold{n}_1,\bold{n}_2\in \mathbb{Z}^N} e^{i\pi  \bold{n}^T_{2}\left(\bold{n}_{1}+\left(\boldsymbol{\mu}+\boldsymbol{\nu}\right)\oslash \bold{d}\right)}
       e^{-\frac{\abs{\sqrt{2\pi}C\boldsymbol{L} - \boldsymbol{\alpha}}^2}{4\tanh \frac{\Delta^2}{2}}}e^{-\frac{\tanh \frac{\Delta^2}{2}}{4}\abs{\sqrt{2\pi}C\boldsymbol{L} + \boldsymbol{\alpha}}^2}
    \end{align}
    where the last equaility follow the same normalization reasoning as the single mode case. Here, the leading order term is an exponential function of the lattice's shortest Euclidean distance, defined formally as 
    \begin{eqnarray}
        \abs{x}_{\text{min}}^2:= \sum_{\bold{n}_1,\bold{n}_2\in \mathbb{Z}^N \notin \set{(\bold{0}, \bold{0})}} \abs{\bold{L}}^2
    \end{eqnarray}
    where $\bold{L}$ is a function of $\bold{n}_{1,2}$.
\end{proof}

With the key lemma extended to the multimode, we are now equipped to derive the QEC matrix expression for the multimode.

\begin{lemma}[Multimode GKP QEC matrix under loss]\label{lem:loss_multimode_QEC}
    For the N mode GKP code, the error overlap matrix
    \begin{align}
        & M_{[\boldsymbol{\mu} \bold{l}],[\boldsymbol{\nu} \bold{k}]} \\
        :=& {}_\Delta\bra{\boldsymbol{\mu}} E_{\bold{l}}^\dagger E_{\bold{k}}\ket{\boldsymbol{\nu}}_\Delta\\
        =& \left(1 +O\left(e^{-\pi\left(n_\Delta + \frac{1}{2}\right)\abs{x}_{\text{min}}^2}\right)\right) \frac{t^{\norm{\frac{\bold{l} + \bold{k}}{2}}_1}}{(\gamma n_\Delta + 1)^N} \sum_{\bold{n}_1,\bold{n}_2\in \mathbb{Z}^N} e^{i\pi  \bold{n}^T_{2}\left(\bold{n}_{1}+\left(\boldsymbol{\mu}+\boldsymbol{\nu}\right)\oslash \bold{d}\right)}e^{-\frac{\pi}{2}\frac{1-\gamma}{ \gamma+\frac{1}{n_{\Delta}}}\left|\bold{L}\right|^2}\left\langle \bold{l}\left|\hat{D}\left({\sqrt{\frac{n_\Delta + 1}{\gamma n_\Delta + 1}}\left(\sqrt{2\pi}C\bold{L}\right)^\ast}\right) \right| \bold{k}\right\rangle
        \label{eq:multimode_QEC}
    \end{align}
    where $\hat{E}_{\bold{l}}:= \otimes_{i=1}^N \hat{E}^{(i)}_{l_i} $, $\left\langle \bold{l}\left|\hat{D}\left(\boldsymbol{\alpha}\right) \right| \bold{k}\right\rangle:= \prod_{i=1}^N \left\langle l_i\left|\hat{D}\left(\alpha_i\right) \right| k_i\right\rangle$, $t := \frac{\gamma n_\Delta}{\gamma n_\Delta + 1}$, and the 1-norm $\norm{\bold{u}} := \sum_{i=1}^N \abs{u_i}$.
\end{lemma}
\begin{proof}
    Extending to multimode, we have that
    \begin{equation}
        \hat{E}_{\bold{l}}^{\dagger} \hat{E}_{\bold{k}}=\bigotimes_{i=1}^N \int \frac{d^2 \alpha_i}{\pi} e^{-\frac{1}{2}(1-\gamma)|\alpha_i|^2}\left\langle l_i\left|\hat{D}(\alpha_i^{\ast})\right| k_i\right\rangle \hat{D}(\alpha_i \sqrt{\gamma}).
    \end{equation}
    Notice that there is a tensor product structure, which makes it easy to work with. 
    \begin{eqnarray}
        && M_{[\boldsymbol{\mu} \bold{l}],[\boldsymbol{\nu} \bold{k}]} \\
        &:=& _{\Delta}\bra{\boldsymbol{\mu}}\hat{E}_{\bold{l}}^{\dagger} \hat{E}_{\bold{k}}\ket{\boldsymbol{{\nu}}}_{\Delta}\\
        &=&\left(1 +O\left(e^{-\frac{\pi}{d}\left(n_\Delta + \frac{1}{2}\right)}\right)\right)\sum_{\bold{n}_1,\bold{n}_2\in \mathbb{Z}^N} e^{i\pi  \bold{n}^T_{2}\left(\bold{n}_{1}+\left(\boldsymbol{\mu}+\boldsymbol{\nu}\right)\oslash \bold{d}\right)} \\
        && \prod_{i=1}^N\int \frac{d^2 \alpha_i}{\pi} e^{-\frac{1-\gamma}{2}\abs{\alpha_i}^2}\bra{l}\hat{D}(\alpha_i^\ast)\ket{k}e^{-\frac{\abs{\sqrt{2\pi}C\bold{L} - \sqrt{\gamma}\alpha_i}^2}{4t_\Delta}}e^{-\frac{t_\Delta}{4}\abs{\sqrt{2\pi}C\bold{L} + \sqrt{\gamma}\alpha_i}^2}.
    \end{eqnarray}
    Each element within the product can be evaluated with the same techniques as we presented in Lemma~\ref{lem:single_mode_QEC}. Therefore, it follows straightforwardly to the final expression.
\end{proof}

The matrix $\bold{L}$ has been appearing in the key Lemmas and it is critical to gain a better understanding of it. Since here we are decomposing most of the operations with displacements, it is natural to imagine that the inner product among GKP codewords is connected to its symplectic dual lattice, similar to the analysis done for GKP under random displacement channels~\cite{PhysRevA.64.062301}. This is indeed the case.

\begin{lemma}[Connection to symplectic dual lattice]\label{lem:dual_lattice_connect}
    Suppose the underlying lattice of the GKP code is $\Lambda$, the lattice is defined as
    \begin{eqnarray}
        \bold{L} := S^{-1}\left(D_{\text{sq}}^{-1/2} \otimes I_2\right)\left(d_1 n_{1,1}+\mu_1-\nu_1, n_{2,1}, \dots, d_N n_{1,N}+\mu_N-\nu_N, n_{2,N}\right)^T\label{eq:symplectic_dual_lattice}
    \end{eqnarray}
    is a sublattice of the symplectic dual lattice. The union of all $\Lambda$ forms the full symplectic dual lattice $\Lambda^\perp$, i.e.
    \begin{eqnarray}
    \Lambda^\perp = \set{\bold{u}\Bigg| \bold{u}\in \boldsymbol{L}, \forall \bold{n}_{1,2}\in \mathbb{Z}^N, \Delta \boldsymbol{\mu}\in\mathbb{Z}^N_{\boldsymbol{d}}}
    \end{eqnarray}
\end{lemma}
\begin{proof}
    The generator matrix of the full lattice can be given as
\begin{eqnarray}
    M^\prime =  \left(D_{\text{sq}}^{-1/2} \otimes I_2\right) S^{-T}
\end{eqnarray}
where as the original GKP lattice is generated by $M = M_\text{sq} S^{-T} = \left(D_{\text{sq}}^{1/2} \otimes I_2\right)  S^{-T}$. One can verify that
\begin{eqnarray}
    M^\prime \Omega M^T =    \left(D_{\text{sq}}^{-1/2} \otimes I_2\right) S^{-T}\Omega S^{-1} \left(D_{\text{sq}}^{1/2} \otimes I_2\right) = \Omega
\end{eqnarray}
where we have used the fact that the inverse of the symplectic matrix $S$ is also a symplectic matrix.
Therefore, the symplectic product of the vectors on the two lattices 
\begin{eqnarray}
        \bold{a}^T M^\prime \Omega M^T \bold{b} = \bold{a}^T \Omega \bold{b} \in \mathbb{Z}
\end{eqnarray}
which is the condition for the symplectic dual lattice.
\end{proof}

\subsection{Extension to amplification channel}

The extension to the amplification channel is straightforward given the displacement operator expansion of the amplification Kraus operators. 

\begin{lemma}[Multimode GKP QEC matrix under amplification]\label{lem:amp_multimode_QEC}
    For the N mode GKP code, the error overlap matrix
    \begin{align}
        & M_{[\boldsymbol{\mu} \bold{l}],[\boldsymbol{\nu} \bold{k}]} \\
        =& \left(1 +O\left(e^{-\pi\left(n_\Delta + \frac{1}{2}\right)\abs{x}_{\text{min}}^2}\right)\right) \frac{t^{\norm{\frac{\bold{l} + \bold{k}}{2}}_1}}{(g m_\Delta + 1)^N} \sum_{\bold{n}_1,\bold{n}_2\in \mathbb{Z}^N} e^{i\pi  \bold{n}^T_{2}\left(\bold{n}_{1}+\left(\boldsymbol{\mu}+\boldsymbol{\nu}\right)\oslash \bold{d}\right)}e^{-\frac{\pi}{2}\frac{1+g}{ g+\frac{1}{m_{\Delta}}}\left|\bold{L}\right|^2}\left\langle \bold{l}\left|\hat{D}\left({\sqrt{\frac{m_\Delta - 1}{g m_\Delta + 1}}\left(\sqrt{2\pi}C\bold{L}\right)^\ast}\right) \right| \bold{k}\right\rangle
        \label{eq:amp_multimode_QEC}
    \end{align}
    where all parameters are defined to be the same as Lemma~\ref{lem:loss_multimode_QEC} except $t := \frac{g m_\Delta}{gm_\Delta + 1}$, $g = G-1$, $m_\Delta = n_\Delta + 1$.
\end{lemma}

The derivation of Lemma~\ref{lem:amp_multimode_QEC} is very similar to the derivation of Lemma~\ref{lem:loss_multimode_QEC} except for the exact decomposition of the noise channel. Therefore, the resulting expression substitutes a few loss-related parameters with amplification-related parameters. As a result, the following derivations for loss can be easily modified to apply to the amplification channel.

\section{Near-optimal performance}

Following the results of the transpose channel, the near-optimal performance of the GKP code can be obtained through its QEC matrix alone. Here, we attempt such a calculation.


\begin{lemma}[Near-optimal performance of the infinite-energy GKP code]\label{lem:inf_energy_performance}
    The near-optimal performance of a GKP code with an underlying symplectic integral lattice, $\Lambda$, can be upper bounded through
    \begin{eqnarray}
        \tilde{\epsilon} \leq \frac{1}{4}\sum_{\bold{x}\in \Lambda^\perp\notin\set{\bold{0}}}e^{-\pi\frac{1-\gamma}{\gamma}\left|\bold{x}\right|^2 }
    \end{eqnarray}
    where $\Lambda^\perp$ is the symplectic dual lattice of $\Lambda$. The logical dimension of the GKP code, $d_L$, is assumed to be finite.
\end{lemma}
\begin{proof}
    From the perturbative expansion of the near-optimal fidelity, we have $ \tilde{\epsilon} = \frac{1}{d} \norm{f(D)\odot \Delta M}_F^2$. Recall that while there are some degrees of freedom in choosing $D$, it should be diagonal, i.e. $\bold{l} = \bold{k}$ and $\boldsymbol{\mu} = \boldsymbol{\nu}$. We can tentatively pick $\bold{n}_1 = \bold{n}_2 = \bold{0}$ to have $\bold{L} = \bold{0}$. With this choice, $D$ has no dependence on $\boldsymbol{\mu}$ and has the form of a thermal distribution
    \begin{eqnarray}
            D_{\bold{l}, \bold{l}} &=& M_{[\boldsymbol{\mu} \bold{l}],[\boldsymbol{\mu} \bold{l}]}\vert_{\bold{n}_{1,2} = \bold{0}}\\
            &=& \frac{1}{(\gamma n_\Delta +1)^N} t^{\norm{\bold{l}}_1}\\
            \Tr D &=& \sum_{\bold{l}\in \mathbb{Z}^{\ast}_N} D_{\bold{l}, \bold{l}}= 1.
    \end{eqnarray}
    where $\mathbb{Z}^\ast = \set{0}\cup \mathbb{Z}^+ $ denotes non-negative integers. Here, we take the prefactor of $1 +O\left(e^{-\pi\left(n_\Delta + \frac{1}{2}\right)\abs{x}_{\text{min}}^2}\right)$ by unity since we are dealing with infinite energy, which is valid for finite $N$ and $d_L$. Such an approximation can be done more rigorously for the asymptotic cases if we focus on specific types of lattices as we do later in Theorem~\ref{theo:rates_GKP_loss_amp}. $\Delta M$ is then the summation except the lattice points included by $D$ such that
    \begin{eqnarray}
            && (f(D)\odot \Delta M)_{[\boldsymbol{\mu} \bold{l}],[\boldsymbol{\nu} \bold{k}]} \\
            &=& \frac{(\gamma n_\Delta + 1)^{-N/2}\sqrt{t}^{\norm{\bold{l} + \bold{k}}_1}}{\sqrt{t}^{\norm{\bold{l}}_1} +\sqrt{t}^{\norm{\bold{k}}_1} }\sum_{\bold{n}_1,\bold{n}_2\in \mathbb{Z}^\prime_N} e^{i\pi  \bold{n}^T_{2}\left(\bold{n}_{1}+\left(\boldsymbol{\mu}+\boldsymbol{\nu}\right)\oslash \bold{d}\right)}e^{-\frac{\pi}{2}\frac{(1-\gamma)}{(\gamma+\frac{1}{n_{\Delta}})}\left|\bold{L}\right|^2}\left\langle \bold{l}\left|\hat{D}\left({\sqrt{\frac{n_\Delta + 1}{\gamma n_\Delta + 1}}\left(\sqrt{2\pi}C\bold{L}\right)^\ast}\right) \right| \bold{k}\right\rangle
    \end{eqnarray}
    where we define $ \mathbb{Z}^{\prime N}$ to be the special set that excludes the term with $\boldsymbol{\mu} = \boldsymbol{\nu}$ and $\bold{n}_1 = \bold{n}_2$. The near-optimal infidelity can be expressed as 
    \begin{eqnarray}
        \tilde{\epsilon}&=& \frac{1}{d_L} \norm{f(D)\odot \Delta M}_F^2\\
        &=& (\gamma n_\Delta + 1)^{-N}\sum_{\boldsymbol{\mu}, \boldsymbol{\nu}\in \mathbb{Z}_{\bold{d}}^N}\sum_{\bold{l}, \bold{k}\in \mathbb{Z}_{+}^N}\sum_{\bold{n}_{1,2},\bold{m}_{1,2}\in \mathbb{Z}^{\prime N}} \frac{t^{\norm{\bold{l} + \bold{k}}_1}}{\left(\sqrt{t}^{\norm{\bold{l}}_1} +\sqrt{t}^{\norm{\bold{k}}_1} \right)^2}e^{i\left(\phi(\boldsymbol{n}_{1,2}, \boldsymbol{\mu}, \boldsymbol{\nu}) - \phi(\boldsymbol{m}_{1,2}, \boldsymbol{\mu}, \boldsymbol{\nu})\right)}\\
        && e^{-\pi\frac{(1-\gamma)}{ (\gamma+\frac{1}{n_{\Delta}})}\left(\left|\bold{L_\bold{n}}\right|^2 + \left|\bold{L}_\bold{m}\right|^2\right)}\left\langle \bold{l}\left|\hat{D}\left(G_\bold{n}\right) \right| \bold{k}\right\rangle \left\langle \bold{k}\left|\hat{D}^\dagger\left(G_\bold{m}\right) \right| \bold{l}\right\rangle\\
        &=& (\gamma n_\Delta + 1)^{-N}\sum_{\boldsymbol{\mu}, \boldsymbol{\nu}\in \mathbb{Z}_{\bold{d}}^N}\sum_{\bold{n}_{1,2},\bold{m}_{1,2}\in \mathbb{Z}^{\prime N}} e^{i\left(\phi(\boldsymbol{n}_{1,2}, \boldsymbol{\mu}, \boldsymbol{\nu}) - \phi(\boldsymbol{m}_{1,2}, \boldsymbol{\mu}, \boldsymbol{\nu})\right)}e^{-\frac{\pi}{2}\frac{(1-\gamma)}{ (\gamma+\frac{1}{n_{\Delta}})}\left(\left|\bold{L_\bold{n}}\right|^2 + \left|\bold{L}_\bold{m}\right|^2\right)}\\
        && \sum_{\bold{l}, \bold{k}\in \mathbb{Z}_{+}^N} \frac{t^{\norm{\bold{l} + \bold{k}}_1}}{\left(\sqrt{t}^{\norm{\bold{l}}_1} +\sqrt{t}^{\norm{\bold{k}}_1} \right)^2}
        \left\langle \bold{l}\left|\hat{D}\left(G_\bold{n}\right) \right| \bold{k}\right\rangle \left\langle \bold{k}\left|\hat{D}^\dagger\left(G_\bold{m}\right) \right| \bold{l}\right\rangle
    \end{eqnarray}
    where we define, for simplicity, $\phi(\boldsymbol{n}_{1,2}, \boldsymbol{\mu}, \boldsymbol{\nu}) = \pi  \bold{n}^T_{2}\left(\bold{n}_{1}+\left(\boldsymbol{\mu}+\boldsymbol{\nu}\right)\oslash \bold{d}\right)$, $G_\bold{n} := \sqrt{\frac{n_\Delta + 1}{\gamma n_\Delta + 1}}\left(\sqrt{2\pi}C\bold{L}_\bold{n}\right)^\ast$, and we use the subscript of $\bold{L}$ to indicate its dependence. Focusing on the last term in the product, we have
    \begin{eqnarray}
    &&  (\gamma n_\Delta + 1)^{-N}\sum_{\bold{l}, \bold{k}\in \mathbb{Z}_{+}^N} \frac{t^{\norm{\bold{l} + \bold{k}}_1}}{\left(\sqrt{t}^{\norm{\bold{l}}_1} +\sqrt{t}^{\norm{\bold{k}}_1} \right)^2}
    \left\langle \bold{l}\left|\hat{D}\left(G_\bold{n}\right) \right| \bold{k}\right\rangle \left\langle \bold{k}\left|\hat{D}^\dagger\left(G_\bold{m}\right) \right| \bold{l}\right\rangle\\
    &=&  (\gamma n_\Delta + 1)^{-N}\sum_{\bold{l}, \bold{k}\in \mathbb{Z}_{+}^N} \frac{1}{\left(t^{\frac{\norm{\bold{l}}_1 - \norm{\bold{k}}_1}{4}} +t^{\frac{\norm{\bold{k}}_1 - \norm{\bold{l}}_1}{4}} \right)^2}t^{\frac{1}{2}\norm{\bold{l} + \bold{k}}_1}
    \left\langle \bold{l}\left|\hat{D}\left(G_\bold{n}\right) \right| \bold{k}\right\rangle \left\langle \bold{k}\left|\hat{D}^\dagger\left(G_\bold{m}\right) \right| \bold{l}\right\rangle\\
    &\leq & (\gamma n_\Delta + 1)^{-N}\frac{1}{4}\sum_{\bold{l}, \bold{k}\in \mathbb{Z}_{+}^N} t^{\frac{1}{2}\norm{\bold{l} + \bold{k}}_1}
    \left\langle \bold{l}\left|\hat{D}\left(G_\bold{n}\right) \right| \bold{k}\right\rangle \left\langle \bold{k}\left|\hat{D}^\dagger\left(G_\bold{m}\right) \right| \bold{l}\right\rangle\label{eq:up_bound_mult_mod}\\
    &=& \frac{1}{4}\prod_{i=1}^N  \frac{1}{\gamma n_\Delta + 1}\sum_{l_i, k_i=0}^\infty t^{\frac{1}{2}(l_i + k_i)}
    \left\langle l_i\left|\hat{D}\left(G_\bold{n}\right)_i \right| k_i\right\rangle \left\langle k_i\left|\hat{D}^\dagger\left(G_\bold{m}\right)_i \right| l_i\right\rangle
    \end{eqnarray}
    where we have used the inequality $x + \frac{1}{x}\geq 2$. For each mode, we can observe that with the summation of $\bold{n}, \bold{m}$, each $l_i, k_i$ is symmetric and can therefore be written as
    \begin{eqnarray}
        && \frac{1}{\gamma n_\Delta + 1}\sum_{l_i, k_i=0}^\infty t^{\frac{1}{2}(l_i + k_i)}
        \left\langle l_i\left|\hat{D}\left(G_\bold{n}\right)_i \right| k_i\right\rangle \left\langle k_i\left|\hat{D}^\dagger\left(G_\bold{m}\right)_i \right| l_i\right\rangle\\
        &=&  2\frac{1}{\gamma n_\Delta + 1}\sum_{\Delta l_i = 0}^\infty t^{\frac{\Delta l_i}{2}}\sum_{l_i=0}^\infty t^{l_i}
        \left\langle l_i\left|\hat{D}\left(G_\bold{n}\right)_i \right| l_i+\Delta l_i\right\rangle \left\langle l_i+\Delta l_i\left|\hat{D}^\dagger\left(G_\bold{m}\right)_i \right| l_i\right\rangle\\
        && - \sum_{l_i=0}^\infty t^{l_i}
        \left\langle l_i\left|\hat{D}\left(G_\bold{n}\right)_i \right| l_i\right\rangle \left\langle l_i\left|\hat{D}^\dagger\left(G_\bold{m}\right)_i \right| l_i\right\rangle\\
        &=& e^{-\left( \abs{G_\bold{n}^i}^2 +\abs{G_\bold{m}^i}^2 \right)\left(\frac{1}{2}+\frac{t}{1-t}\right)} \left(2\sum_{\Delta l_i=0}^\infty t^{\Delta l_i/2}e^{i\Delta l_i\left(\theta\left(G_\bold{n}^i\right) - \theta\left(G_\bold{m}^i\right)\right)} I_{\Delta l_i}\left(2\frac{\abs{G_\bold{n}^iG_\bold{m}^i}\sqrt{t}}{1 - t}\right) - I_{0}\left(2\frac{\abs{G_\bold{n}^iG_\bold{m}^i}\sqrt{t}}{1 - t}\right)\right)
    \end{eqnarray}
    where $\Delta l_i := k_i - l_i$. The second equality is a result of the Hardy-Hill formula,
    \begin{eqnarray}
            \sum_{n=0}^\infty \frac{n!}{\Gamma\left(1 + \alpha + n\right)}L_n^\alpha(x)L_n^\alpha(y) t^n = \frac{1}{(xyt)^{\alpha/2}(1-t)}e^{-(x+y)t/(1-t)}I_\alpha\left(\frac{2\sqrt{xyt}}{1-t}\right)
    \end{eqnarray}
    with $I$ being the modified Bessel function of the first kind. The infinite-energy assumption comes in to greatly simplifies the expression by eliminating the cross terms. The second term vanishes in the limit of infinite energy since $I_\alpha(x)\to \frac{e^x}{\sqrt{2\pi x}}$ when $x\gg \alpha$. Focusing on the first term, we have
    \begin{eqnarray}
        && 2e^{-\left( \abs{G_\bold{n}^i}^2 +\abs{G_\bold{m}^i}^2 \right)\left(\frac{1}{2}+\frac{t}{1-t}\right)}\sum_{\Delta l_i = 0}^\infty  t^{\Delta l_i}e^{i\Delta l_i\left(\theta\left(G_\bold{n}^i\right) - \theta\left(G_\bold{m}^i\right)\right)} I_{\Delta l_i}\left(2\frac{\abs{G_\bold{n}^iG_\bold{m}^i}\sqrt{t}}{1 - t}\right)\\
        &=&2e^{-\left( \abs{G_\bold{n}^i}^2 +\abs{G_\bold{m}^i}^2 \right)\left(\frac{1}{2}+\frac{t}{1-t}\right)}\sum_{\Delta l_i = 0}^\infty  t^{\Delta l_i} I_{\Delta l_i}\left(2\frac{\abs{G_\bold{n}^iG_\bold{m}^i}\sqrt{t}}{1 - t}\right)\\
        &=&e^{-\left( \abs{G_\bold{n}^i}^2 +\abs{G_\bold{m}^i}^2 \right)\left(\frac{1}{2}+\frac{t}{1-t}\right)} \exp{2\frac{\abs{G_\bold{n}^iG_\bold{m}^i}\sqrt{t}}{1 - t}}\\
        &=&e^{-\left( \abs{G_\bold{n}^i}^2 +\abs{G_\bold{m}^i}^2 \right)\left(\frac{1}{2}+\frac{1-\epsilon}{\epsilon}\right) + 2\frac{\abs{G_\bold{n}^iG_\bold{m}^i}\sqrt{1-\epsilon}}{\epsilon}} \\
        &=&e^{-\left( \abs{G_\bold{n}^i}^2 +\abs{G_\bold{m}^i}^2 - 2 \abs{G_\bold{n}^iG_\bold{m}^i}\right)\frac{1}{\epsilon} + \frac{1}{2}\left( \abs{G_\bold{n}^i}^2 +\abs{G_\bold{m}^i}^2 - 2 \abs{G_\bold{n}^iG_\bold{m}^i}\right) + O(\epsilon)}\\
        &=& \delta(G_\bold{n}^i - G_\bold{m}^i)
    \end{eqnarray}
    where $\epsilon:=1-t\to0$ at infinite energy. The first equality is taken at $\theta\left(G_\bold{n}^i\right) = \theta\left(G_\bold{m}^i\right)$ since otherwise the summation of the complex phases would vanish. In the second equality, we used the mathematical identity of $\sum_{i=0}^\infty (1-\frac{1}{x})^{ci} I_i(x)\to \frac{1}{2}e^x$ as $x\to\infty$ for constant $c$, which is proved in Lemma~\ref{lem:Bessel_sum}. The final result implies that $G_\bold{n}^i = G_\bold{m}^i$, which leads to $\bold{n} = \bold{m}$. Therefore, in the following, the complex phase we haven't taken into consideration, $\phi(\boldsymbol{n}_{1,2}, \boldsymbol{\mu}, \boldsymbol{\nu}) - \phi(\boldsymbol{m}_{1,2}, \boldsymbol{\mu}, \boldsymbol{\nu})$ also vanishes. Combining these with the full expression, we have that 
    \begin{eqnarray}
        \tilde{\epsilon} &\leq& \frac{1}{4d_L}\sum_{\boldsymbol{\mu}, \boldsymbol{\nu}\in \mathbb{Z}_{\bold{d}}^N}\sum_{\bold{n}_{1,2}\in \mathbb{Z}^{\prime N}}e^{-\pi\frac{1-\gamma}{\gamma}\left|\bold{L_\bold{n}}\right|^2 }\\
        &=& \frac{1}{4d_L}\sum_{\boldsymbol{\mu}, \boldsymbol{\nu}\in \mathbb{Z}_{\bold{d}}^N}\sum_{\bold{n}_{1,2}\in \mathbb{Z}^{\prime N}}e^{-\pi\frac{1-\gamma}{\gamma}\left|\bold{L_\bold{n}}\right|^2 }\\
        &=& \frac{1}{4}\sum_{\Delta_\mu \in \mathbb{Z}_+^\bold{d}}\sum_{\bold{n}_{1,2}\in \mathbb{Z}^{\prime N}}e^{-\pi\frac{1-\gamma}{\gamma}\left|\bold{L_\bold{n}}\right|^2 }\\
        &=& \frac{1}{4}\sum_{\bold{x}\in \Lambda^\perp\notin\set{\bold{0}}}e^{-\pi\frac{1-\gamma}{\gamma}\left|\bold{x}\right|^2 }\\
    \end{eqnarray}
    where we applied the multiplicity of $\sum_{\mu_i, \nu_i=1}^{d_i} = d_i\sum_{\Delta \mu_i = 0}^{d_i-1}$, the total logical dimension $d_L = \prod_{i=1}^N d_i$, and the connection between $\bold{L}$ and the symplectic dual lattice $\Lambda^\perp$ given in Lemma~\ref{lem:dual_lattice_connect}.
\end{proof}

While in the multimode case, one can only obtain an upper bound, in single mode we can obtain the exact expressions of the near-optimal performance expression. The key difference between single mode and multimode is that instead of resorting to obtaining an upper bound of the near-optimal infidelity as in Eq.~\eqref{eq:up_bound_mult_mod}, we have
\begin{eqnarray}
    && (\gamma n_\Delta + 1)\sum_{l,k\in \mathbb{Z}_{+}} \frac{t^{l+k}}{\left(\sqrt{t}^{l} +\sqrt{t}^k \right)^2}
    \left\langle l\left|\hat{D}\left(G_n\right) \right| k\right\rangle \left\langle k\left|\hat{D}^\dagger\left(G_m\right) \right| l\right\rangle\\
    &=& 2 (\gamma n_\Delta + 1)\sum_{\Delta l = 0}^\infty \frac{t^{\Delta l}}{(\sqrt{t}^{\Delta l} + 1)^2} \sum_{l\in \mathbb{Z}_{+}} t^l \left\langle l\left|\hat{D}\left(G_n\right) \right| l+\Delta l\right\rangle \left\langle l+\Delta l\left|\hat{D}^\dagger\left(G_m\right) \right| l\right\rangle\\
    && -  \frac{1}{4} (\gamma n_\Delta + 1)\sum_{l\in \mathbb{Z}_{+}} t^l \left\langle l\left|\hat{D}\left(G_\bold{n}\right) \right| l\right\rangle \left\langle l\left|\hat{D}^\dagger\left(G_\bold{m}\right) \right| l\right\rangle\\
    &=& e^{-\left( \abs{G_n}^2 +\abs{G_m}^2 \right)\left(\frac{1}{2}+\frac{t}{1-t}\right)} \left( 2 \sum_{\Delta l = 0}^\infty \frac{\sqrt{t}^{\Delta l}}{(\sqrt{t}^{\Delta l} + 1)^2} e^{i\Delta l\left(\theta\left(G_n\right) - \theta\left(G_m\right)\right)} I_{\Delta l}\left(2\frac{\abs{G_nG_m}\sqrt{t}}{1 - t}\right) - I_{0}\left(2\frac{\abs{G_n G_m}\sqrt{t}}{1 - t}\right)\right)
    \end{eqnarray}
    Therefore, at infinite energy, we know precisely the expression instead of the upper bound as 
    \begin{eqnarray}
        \lim_{n_\Delta \to \infty} \tilde{\epsilon}&=& \frac{1}{4}\sum_{\bold{x}\in \Lambda^\perp\notin\set{\bold{0}}}e^{-\pi\frac{1-\gamma}{\gamma}\left|\bold{x}\right|^2 }.
    \end{eqnarray}
    At finite energy, we have the full expression as
    \begin{eqnarray}
         \tilde{\epsilon} &=& \frac{1}{d_L}\sum_{\mu, \nu\in \mathbb{Z}_{\bold{d}}}\sum_{n_{1,2},m_{1,2}\in \mathbb{Z}} e^{i\left(\phi(n_{1,2}, \mu, \nu) - \phi(m_{1,2}, \mu, \nu)\right)}k^{\abs{L_n}^2 + \abs{L_m}^2}\times \nonumber\\
        &&  \left( 2 \sum_{\Delta l = 0}^\infty \frac{\sqrt{t}^{\Delta l}}{(\sqrt{t}^{\Delta l} + 1)^2} e^{i\Delta l\left(\theta\left(L_n\right) - \theta\left(L_m\right)\right)} I_{\Delta l}\left(z\right) - I_{0}\left(z\right)\right)\label{eq:finite_energy_single_mode}
    \end{eqnarray}
    where $k := e^{-\pi \left(n_\Delta + \frac{1}{2}\right)}$ and $z := 2\frac{\abs{G_nG_m}\sqrt{t}}{1 - t} = \pi (n_\Delta + 1) \sqrt{t}\abs{L_n L_m}$.

One remark we give is that here we are ignoring the overlap between the codewords, which is exponentially suppressed by $n_\Delta$. For instance, one can check that the finite energy performance in Eq.~\eqref{eq:finite_energy_single_mode} at $\gamma = 0$ is exponentially suppressed but finite, which is precisely because we ignored the overlap. If we include the overlap, the near-optimal infidelity should vanish when the noise channel is an identity channel. 

To give two concrete examples, we consider the square lattice and the hexagonal lattice. For a square lattice encoding $d$-dimensional logical space, the generator is
\begin{eqnarray}
        M_{\text{sq}} = \sqrt{d}\begin{pmatrix}
            1 & 0\\
            0 & 1
        \end{pmatrix}.
\end{eqnarray}
Suppose we only consider the closest lattice points of the symplectic dual lattice, which has a Euclidean distance given by $\abs{x}^2_{\text{min}} = \frac{1}{d}$. By enumerating all values of $n_{1,2}, m_{1,2}$ and $\mu, \nu$, we can count the number of points that contribute to the leading order term with their relative phase and obtain the performance to be 
\begin{eqnarray}
     1 - \tilde{F}^\text{opt} &\approx& 2 e^{-\frac{2\pi}{d} \left(n_\Delta + \frac{1}{2}\right)} \left( 8\sum_{\substack{\Delta l\in \mathbb{Z}^\ast \\ \Delta l=0(\mathrm{mod}\ 2) }}\frac{\sqrt{t}^{\Delta l}}{(\sqrt{t}^{\Delta l} + 1)^2}I_{\Delta l}(z_{\text{sq}}) - I_0(z_{\text{sq}}) -1\right)\\
     &\to& e^{-\frac{\pi}{d}\frac{1-\gamma}{\gamma}} \text{ at }n_\Delta \to \infty
\end{eqnarray}
with $z_{\text{sq}} := \frac{\pi}{d} (n_\Delta + 1) \sqrt{\frac{\gamma n_\Delta}{\gamma n_\Delta + 1}}$. The last term in the finite energy expression is due to the codeword overlap. The infinite energy expression can be read from the closest-point distance and has 4 closest points.

Similarly, we can compute the expression explicitly for the hexagonal lattice through its generating matrix and the symplectic transformation required to transform it into the square lattice.
\begin{eqnarray}
    M_{\text{hex}} &=& \sqrt{\frac{2d}{\sqrt{3}}}\begin{pmatrix}
        1 & 0\\
        -\frac{1}{2} & \frac{\sqrt{3}}{2}
    \end{pmatrix},\\
    S_{\text{hex}} &=& \frac{1}{3^{1/4}}\begin{pmatrix}
        \sqrt{\frac{3}{2}} & \frac{1}{\sqrt{2}}\\
        0 & \sqrt{2}
    \end{pmatrix}.
\end{eqnarray}
As a result, we get the performance as 
\begin{eqnarray}
    1 - \tilde{F}^\text{opt} &\approx& 3 e^{-\frac{4\pi}{\sqrt{3}d} \left(n_\Delta + \frac{1}{2}\right)} \left( 8\sum_{\substack{\Delta l\in \mathbb{Z}^\ast \\ \Delta l=0(\mathrm{mod}\ 2) }}\frac{\sqrt{t}^{\Delta l}}{(\sqrt{t}^{\Delta l} + 1)^2}I_{\Delta l}(z_{\text{hex}}) - I_0(z_{\text{hex}}) -1\right)\\
    &\to& \frac{3}{2}e^{-\frac{2\pi}{\sqrt{3}d}\frac{1-\gamma}{\gamma}} \text{ at }n_\Delta \to \infty
\end{eqnarray}
where $z_{\text{hex}} := \frac{2\pi}{\sqrt{3}d} (n_\Delta + 1) \sqrt{\frac{\gamma n_\Delta}{\gamma n_\Delta + 1}}$. Notice that the exponents and the prefactors have changed compared to the square lattice. This is a direct consequence of the closest-point distance increasing to $\abs{x}^2_{\text{min}} = \frac{2}{\sqrt{3}}\frac{1}{d}$ and having 6 closest points.

\section{Achievable rate}

The achievable rate defines the maximum rate of information transmission with vanishing error probability. Equivalently, we can obtain the achievable rate by requiring the entanglement infidelity to vanish at an infinite system size limit. To obtain the achievable rate of the GKP code, we pick a self-dual lattice and apply Lemma.~\ref{lem:Busar_good_lattice}.

\begin{theorem}[Achievable rate of the GKP code under pure loss and amplification]\label{theo:rates_GKP_loss_amp}
    The achievable rates of the GKP code under loss channel with loss rate $\gamma$ and amplification channel with gain $G$ are 
    \begin{eqnarray}
            R_{\textrm{loss}} &=& \log_2\frac{1-\gamma}{\gamma},\\
            R_{\textrm{amp}} &=& \log_2\frac{G}{G-1}.
    \end{eqnarray}
\end{theorem}
\begin{proof}
    If we first focus on loss, from Lemma.~\ref{lem:inf_energy_performance} we already obtained the performance of the infinite-energy GKP as 
    \begin{eqnarray}
        \tilde{\epsilon} \leq \frac{1}{4}\sum_{\bold{x}\in \Lambda^\perp\notin\set{\bold{0}}}e^{-\pi\frac{1-\gamma}{\gamma}\left|\bold{x}\right|^2 }.
    \end{eqnarray}
    Considering a symplectic self-dual lattice $\Lambda_0$ generated by $M_0$. While the symplectic self-dual lattice cannot encode information since $\text{det}M_0 = 1$, we can rescale the generator matrix as $M = \sqrt{\lambda} M_0$ such that the generated lattice $\Lambda $ encodes $\lambda^N$-dimensional logical information. As a result, the sympelctic dual lattice scales as $\Lambda^\perp = \frac{1}{\sqrt{\lambda}}\Lambda_0^\perp = \frac{1}{\sqrt{\lambda}}\Lambda_0$. Given by Lemma.~\ref{lem:Busar_good_lattice}, there exists a good symplectic self-dual lattice $\Lambda_0$ such that 
    \begin{eqnarray}
        \tilde{\epsilon} &\leq& \frac{1}{4}\sum_{\bold{x}\in \Lambda^\perp\notin\set{\bold{0}}}e^{-\pi\frac{1-\gamma}{\gamma}\left|\bold{x}\right|^2 }\\
        &=& \frac{1}{4}\sum_{\bold{x}\in \Lambda_0\notin\set{\bold{0}}}e^{-\frac{\pi}{\lambda}\frac{1-\gamma}{\gamma}\left|\bold{x}\right|^2 }\\
        &\leq& \frac{1}{4}\int e^{-\frac{\pi}{\lambda}\frac{1-\gamma}{\gamma}\abs{\bold{x}}^2}d^{2N} \bold{x}\\
        &=& \frac{1}{4}\left(\lambda \frac{\gamma}{1-\gamma}\right)^N.
    \end{eqnarray}
    Therefore, the near-optimal infidelity vanishes at infinite system size limit $N\to\infty$ if and only if $\lambda \leq \frac{1-\gamma}{\gamma}$. By Corollary~\ref{cor:achieve_rate_near_optimal_fid}, this is sufficient for vanishing optimal channel infidelity. Therefore, this constitutes a lower bound of the GKP code achievable rate, $R\geq \log_2 \lambda$. However, since it coincides with the channel capacity of the loss channel, which is the supremum of all achievable rates, the GKP code exactly achieves the channel capacity,
    \begin{eqnarray}
        R_{\textrm{loss}} &=& \log_2\frac{1-\gamma}{\gamma}.
    \end{eqnarray}
    Moreover, while it can be more rigorously derived, we infer here the GKP's performance through an observation for simplicity. Comparing the multimode QEC matrix expression given in Lemma.~\ref{lem:loss_multimode_QEC} and \ref{lem:amp_multimode_QEC}, we remark that when $n_\Delta, m_\Delta \to \infty$, the only difference between the two is that the exponential function is modified from $e^{-\frac{\pi}{2}\frac{1-\gamma}{ \gamma}\left|\bold{L}\right|^2}$ to $e^{-\frac{\pi}{2}\frac{G}{ G-1}\left|\bold{L}\right|^2}$. Therefore, with similar derivations, one can show 
    \begin{eqnarray}
        R_{\textrm{amp}} &=& \log_2\frac{G}{G-1}.
    \end{eqnarray}
    One factor we neglected is the normalization factor. In Lemma~\ref{lem:inf_energy_performance}, we derived the upper bound assuming finite $d_L$ to approximate the normalization factor by unity, which does not hold when we consider asymptotic quantities like the achievable rates. However, with a given set of lattice-like the self-dual lattice, we can still show this property holds rigorously,
    \begin{eqnarray}
        N^2_{\boldsymbol{\mu}}(\frac{1}{2\sqrt{\pi}(1-e^{-2\Delta^2})})^N -1 &=& \sum_{\boldsymbol{n}_{1,2}, \boldsymbol{m}_{1,2} \in \mathbb{Z}^N} e^{i\left(\phi(\boldsymbol{n}_{1,2}, \boldsymbol{\mu}, \boldsymbol{\mu}) - \phi(\boldsymbol{m}_{1,2}, \boldsymbol{\mu}, \boldsymbol{\mu})\right)}e^{-\pi (n_\Delta + \frac{1}{2})\abs{\bold{L}}^2}\\
        &\leq & \sum_{\bold{x}\in \Lambda^\perp\notin\set{\bold{0}}} e^{-\pi (n_\Delta + \frac{1}{2})\abs{\bold{x}}^2}\\
        &\leq&  \left(\lambda \frac{1}{n_\Delta + \frac{1}{2}}\right)^N.
    \end{eqnarray}
    While $n_\Delta\to \infty$, the symplectic lattice that achieves capacity has a finite $\lambda$. Therefore, the normalization constant $N^2_{\boldsymbol{\mu}} = (2\sqrt{\pi}(1-e^{-2\Delta^2}))^N$ in this limit and justifies our choice in deriving Lemma~\ref{lem:inf_energy_performance} for the symplectc lattice of interest here.
\end{proof}

\section{Mathematical identities}

\begin{lemma}[Summation of the modified Bessel function]\label{lem:Bessel_sum}
    The following identity holds,
     \begin{eqnarray}
         \lim_{x\to \infty}e^{-x}\sum_{n=0}^\infty (1-\frac{1}{x})^{cn}I_{n}(x) = \frac{1}{2},\label{eq:ap_lemma}
     \end{eqnarray}
     where $I_i(x)$ is modified Bessel function, and $c$ is a constant.
 \end{lemma}
 \begin{proof}
     We first prove that 
     \begin{eqnarray}
         \lim_{x\to \infty}e^{-x}\sum_{n=0}^\infty I_{n}(x) = \frac{1}{2}\label{eq:sum_mod_bes}.
     \end{eqnarray}
     The generating function of the modified Bessel function reads
     \begin{align}
         \sum_{i=-\infty}^{\infty} t^n I_n(x) = e^{\frac{x}{2}(t+t^{-1})},
     \end{align}
     where we let $t=1$ and make use of $I_{-n}(x) = I_{n}(x)$ for $i \in \mathbb{Z}$, then we have
     \begin{align}
         e^{x} = -I_0(x) + 2(\sum_{i=0}^{\infty}I_n(x)).
     \end{align}
     Because the asymptotic expansion of $I_0(x)$ reads
     \begin{align}
         I_\alpha (x) \sim \frac{e^{x}}{\sqrt{2\pi x}} \qquad |\arg x| < \frac{\pi}{2}, \label{eq:asymp}
     \end{align}
     we have
     \begin{eqnarray}
         \lim_{x\to \infty}e^{-x}\sum_{n=0}^\infty I_{n}(x) = \frac{1}{2}-\lim_{x\to \infty}\frac{e^x}{\sqrt{2\pi x}}e^{-x} = \frac{1}{2}.
     \end{eqnarray}
     Thus we prove Eq.~\eqref{eq:sum_mod_bes}. Compare Eq.~\eqref{eq:sum_mod_bes} and Eq.~\eqref{eq:ap_lemma}, we find that it is sufficient to prove
     \begin{align}
           \lim_{x\to\infty}  \sum_{n=0}^{\infty}(1-(1-\frac{1}{x})^{cn})e^{-x}I_n(x) = 0.
     \end{align}
     Since $I_0 (x)\geq 0$ for $x\geq 0$, the summation is non-negative. With Bernoulli's inequality states $(1-\frac{1}{x})^{cn} \ge 1-\frac{cn}{x}$ for $x > 1$, we have
     \begin{eqnarray}
         e^{-x}\sum_{i=0}^{\infty}(1-(1-\frac{1}{x})^{cn})I_n(x) &\le& e^{-x}\sum_{n=0}^{\infty}\frac{cn}{x}I_n(x)\\
         &\le& e^{-x}\frac{c}{2}\sum_{i=1}^{\infty}(I_{n-1}(x)-I_{n+1}(x))\\
         &\le& \frac{c}{2} e^{-x}(I_{-1}(x) +I_{0}(x) - 2\lim_{n\to\infty}I_{n}(x)) \\
     &\le& \frac{c}{2} e^{-x}(I_{-1}(x) +I_{0}(x)),
     \end{eqnarray}
     where we have used the recurrence relation of the modified Bessel function, $\frac{2n}{x}I_n(x) = I_{n-1}(x) - I_{n+1}(x)$. Thus we have
     \begin{align}
         \lim _{x \rightarrow \infty} \sum_{n=0}^{\infty}\left(1-\left(1-\frac{1}{x}\right)^{c n}\right)e^{-x} I_n(x) \le \lim _{x \rightarrow \infty} \frac{c}{2} e^{-x}(I_{-1}(x) +I_{0}(x)) = 0,
     \end{align}
     where we use the asymptotic behavior Eq.~\eqref{eq:asymp}
 \end{proof}

    \section{QEC matrix of general codes under loss and amplification}
    The establishment of GKP's rate is based on the existence of a self-dual symplectic lattice with good packing properties. As soon as one such good lattice, $\Lambda_0$ is found, all the capacity-achieving GKP lattices are obtained by scaling it by a factor dependent on the noise strength. One observation is that if we define 
    \begin{eqnarray}
        K_{\textrm{loss}}:=2^{R_{\textrm{loss}}} &=& \frac{1-\gamma}{\gamma}\\
        K_{\textrm{amp}}:=2^{R_{\textrm{amp}}} &=& \frac{G}{G-1},
    \end{eqnarray}
    then when 
    \begin{eqnarray}
        K_{\textrm{loss}} = K_{\textrm{amp}}\label{eq:equivalence_loss_amp}
    \end{eqnarray}
    the same GKP lattice achieves the capacity of loss and amplification simultaneously. Note that here while the encoding is the same, the recovery is tailored to loss or amplification. Such a result has its roots in similar forms of loss and amplification. In particular, we have that
    \begin{eqnarray}
        \hat{E}_{l}^\dagger \hat{E}_k &=& \frac{1}{\pi\left(K_{\textrm{loss}} + 1\right)}\int d^2 \beta e^{-\frac{1}{2}K_{\textrm{loss}}\abs{\beta}^2}\sandwich{l}{\hat{D}(\sqrt{K_{\textrm{loss}}+1}\beta^\ast)}{k}\hat{D}_{\beta}\\
        \hat{A}_{l}^\dagger \hat{A}_k &=& \frac{1}{\pi\left(K_{\textrm{amp}} - 1\right)}\int d^2 \beta e^{-\frac{1}{2} K_{\textrm{amp}}\abs{\beta}^2}\sandwich{l}{\hat{D}(\sqrt{ K_{\textrm{amp}}-1}\beta^\ast)}{k}\hat{D}_{\beta}
    \end{eqnarray}
    where we simply rewrite the expressions obtained in Appendix.~\ref{sec:appendix_noise_channel}. The similarities between the Kraus operator combinations are apparent in the exponential factor. However, there are other components that are not identical. Here, we show that when we evaluate the code with some specific metric, namely $\norm{\Delta M}_F^2$, general codes have the same performance.
    \begin{lemma}[Equivalence of general codes under loss and amplification]\label{lem:KL_cond_converge}
        For any code defined by a set of codewords $\ket{\mu}$ with $\mu\in [d_L]$, its deviation from the Knill-Laflamme condition, characterized by $\norm{\Delta M}_F^2$, is equal if 
        \begin{eqnarray}
            K_{\textrm{loss}} = K_{\textrm{amp}}.
        \end{eqnarray}
     \end{lemma}
     \begin{proof}
        Elements of the QEC matrix are defined as $M_{[\mu l], [\nu k]} = \bra{\mu}\hat{E}_{l}^\dagger \hat{E}_k\ket{\nu}$ for loss and $M_{[\mu l], [\nu k]} = \bra{\mu}\hat{A}_{l}^\dagger \hat{A}_k\ket{\nu}$ for amplification. Suppose $K_{\textrm{loss}} = K_{\textrm{amp}} = K$, we can rewrite the QEC matrix elements as
    \begin{eqnarray}
        M_{[\mu l], [\nu k]} = \bra{l} \frac{1}{\pi (K\pm 1)}\int d^2 \beta \chi(\beta) \hat{D}(\sqrt{K \pm 1}\beta^\ast)\ket{k}.
    \end{eqnarray}
    where the plus and minus sign correspond to loss and amplification respectively, and $\ket{l}$ and $\ket{k}$ are Fock states. Here, we define $\chi(\beta):= e^{-\frac{1}{2}K\abs{\beta}^2}\bra{\mu}\hat{D}_{\beta}\ket{\nu}$, which can be seen as the scaled characteristic function of the operator in between. For $\Delta M:= M - I_L \otimes \left(\Tr_L M\right)$, we can define a modified characteristic function as
    \begin{eqnarray}
        \Delta \chi(\beta) = 
        \begin{cases}
                e^{-\frac{1}{2}K\abs{\beta}^2}\bra{\mu}\hat{D}_{\beta}\ket{\nu} & \text{if } \mu\neq\nu\\
                e^{-\frac{1}{2}K\abs{\beta}^2}\left(\bra{\mu}\hat{D}_{\beta}\ket{\mu}-\frac{1}{d_L}\sum_{\xi}\bra{\xi}\hat{D}_{\beta}\ket{\xi}\right) & \text{if } \mu=\nu
              \end{cases} 
    \end{eqnarray}
    and, for simplicity, the operator $\hat{O}_{\mu, \nu}:= \frac{1}{\pi (K\pm 1)}\int d^2 \beta \Delta\chi(\beta) \hat{D}(\sqrt{K\pm1}\beta^\ast)$. Then, we can express $(\Delta M)_{[\mu l], [\nu k]} = \bra{l}O_{\mu,\nu}\ket{k}$. As a result, we have
    \begin{eqnarray}
        \norm{\Delta M}_F^2 &=&  \sum_{\mu, \nu}\sum_{l,k}\abs{\bra{\mu^{(n)}}\hat{E}_{l}^\dagger \hat{E}_k\ket{\nu^{(n)}}}^2\\
               &=& \sum_{\mu, \nu}\sum_{l,k}\bra{k}O^\dagger_{\mu,\nu}\ket{l}\bra{l}O_{\mu,\nu}\ket{k}\\
               &=& \sum_{\mu, \nu}\sum_{k}\bra{k}O^\dagger_{\mu,\nu}O_{\mu,\nu}\ket{k}\\
               &=& \sum_{\mu, \nu}\Tr{O^\dagger_{\mu,\nu}O_{\mu,\nu}}\\
               &=& \frac{1}{\pi^2}\sum_{\mu, \nu}\int d^2 \beta\abs{\Delta \chi(\beta)}^2.
    \end{eqnarray}
    Note that the final expression of $\norm{\Delta M}_F^2$ is independent of whether the noise channel is loss or amplification.
    \end{proof}
    The above lemma does not utilize any structure in the code itself. However, it also does not directly imply results on the error rate or achievable rate of the code. Past works~\cite{PhysRevLett.104.120501, zhao2023extractingerrorthresholdsframework} have examined the connection between quantum relative entropy, Bures metric, and the decoding error probability. While it is possible to show that $\norm{\Delta M}_F^2$ is a necessary condition, it is unknown whether it is sufficient. Given the almost identical structure of loss and amplification, we conjecture that if a general code can achieve a rate of $R$ under pure loss channel, the same code can achieve the same rate under amplification channel, provided that Eq.~\eqref{eq:equivalence_loss_amp} holds. We note that connections between loss and amplification channels have been explored in other contexts, such as degradability~\cite{PhysRevA.74.062307}. We leave the exploration along this direction to future work.

\end{document}